\newcommand{\typeof}{1} %

\RequirePackage{ifthen}

\newcommand{\condinc}[2]{\ifthenelse{\equal{\typeof}{0}}{#1}{#2}}
\newcommand{\longshortv}[2]{\ifthenelse{\equal{\typeof}{0}}{#1}{#2}}

\documentclass[numbers]{sigplanconf}

\usepackage{balance}

\usepackage{microtype}
\usepackage{bm}
\usepackage{amsmath}
\usepackage{url}
\usepackage{amsthm}
\usepackage{color}
\usepackage{graphicx}
\usepackage{mathrsfs}
\usepackage{amssymb}
\usepackage{cmll}
\usepackage{stmaryrd}
\usepackage{proof}
\usepackage{wrapfig}

\theoremstyle{plain}
\newtheorem{theorem}{Theorem}
\newtheorem{proposition}[theorem]{Proposition}
\newtheorem{lemma}[theorem]{Lemma}
\newtheorem{corollary}[theorem]{Corollary}

\theoremstyle{definition}
\newtheorem{definition}[theorem]{Definition}
\newtheorem{fact}[theorem]{Fact}
\newtheorem{example}[theorem]{Example}

\newtheorem{remark}[theorem]{Remark}

\newenvironment{varitemize}
{
\begin{list}{\labelitemi}
{\setlength{\itemsep}{0pt}
 \setlength{\topsep}{0pt}
 \setlength{\parsep}{0pt}
 \setlength{\partopsep}{0pt}
 \setlength{\leftmargin}{15pt}
 \setlength{\rightmargin}{0pt}
 \setlength{\itemindent}{0pt}
 \setlength{\labelsep}{5pt}
 \setlength{\labelwidth}{10pt}
}}
{
 \end{list} 
}

\newcounter{number}

\newenvironment{varenumerate}
{\begin{list}{\arabic{number}.}
  {
   \usecounter{number}
   \setlength{\labelwidth}{4.0mm}
   \setlength{\labelsep}{2.0mm}
   \setlength{\itemindent}{0.0mm}
   \setlength{\itemsep}{0.0mm}
   \setlength{\topsep}{0.0mm}
   \setlength{\parskip}{0.0mm}
   \setlength{\parsep}{0.0mm}
   \setlength{\partopsep}{0.0mm}
  }
}
{\end{list}}

\newcommand{\tens}{\otimes}
\newcommand{\lin}{\multimap}

\renewcommand{\b}{{}^{\bot}}

\newcommand{\PSIAM}{{\rm\textsf{MSIAM}}}
\newcommand{\SIAM}{{\rm\textsf{SIAM}}}

\newcommand{\ML}{\ensuremath{\mathsf{ML}}}
\newcommand{\PCF}{\ensuremath{\mathsf{PCF}^{\mathsf{LL}}}}
\newcommand{\PCFwo}{\ensuremath{\mathsf{PCF}}}

\newcommand{\MELL}{\textsf{MELL}}

\newcommand{\SMELLY}{\textsf{SMEYLL}}

\newcommand{\one}{\mathsf{1}}
\newcommand{\botlk}{\mathsf{bot}}
\newcommand{\onelk}{\mathsf{one}}

\newcommand{\dlk}{\mathsf{?d}}

\newcommand{\up}{\uparrow}
\newcommand{\down}{\downarrow}
\newcommand{\stable}{\leftrightarrow}
\newcommand{\dr}{\mathrm{dir}}

\newcommand{\atomone}{\alpha}

\newcommand{\formone}{A}
\newcommand{\formtwo}{B}

\newcommand{\bnf}{\mathrel{::=}}
\newcommand{\midd}{\mid}

\newcommand{\posfone}{P}
\newcommand{\negfone}{N}

\newcommand{\netone}{R}

\newcommand{\rednet}{\rightsquigarrow}

\newcommand{\ie}{\emph{i.e.}}
\newcommand{\eg}{\emph{e.g.}}

\newcommand{\ONES}{\mathtt{ONES}}

\newcommand{\DEREL}{\mathtt{DER}}
\newcommand{\DER}{\mathtt{DER}}
\newcommand{\START}{\mathtt{START}}
\newcommand{\PDOORS}{\mathtt{STABLE}}
\newcommand{\id}{\mathtt{Copies}_{\st}}
\newcommand{\POSALL}{\mathtt{POS}} 
\newcommand{\POSI}{\mathtt{INIT}}
\newcommand{\POSF}{\mathtt{FIN}}

\usepackage{xcolor}

\newcommand{\BB}{b}

\newcommand{\bbox}{$\bot$-box}

\newcommand{\st}{\mathbf{T}}

\newcommand{\sttwo}{\mathbf{U}}

\newcommand{\pp}{\mathbf{p}}

\renewcommand{\ss}{\mathbf{s}}

\renewcommand{\b}{{}^{\bot}}

\newcommand{\encode}[2]{\lceil #1, #2 \rceil}

\newcommand{\stk}{\mathord{s}}
\newcommand{\bstk}{\mathord{t}}
\newcommand{\edg}{\mathord{e}}

\newcommand{\emp}{\epsilon}

\newcommand{\M}{\mathcal{M}}
\newcommand{\machine}[1]{\M_{#1}}
\newcommand{\redsiam}{\rightarrow}
\newcommand{\stopsiam}{\nrightarrow}
\newcommand{\sts}{\mathcal{S}}
\newcommand{\trsf}{\mathrm{trsf}}
\newcommand{\orig}{\mathrm{orig}}
\newcommand{\partto}{\rightharpoonup}

\newcommand{\Had}{\mathsf{H}}
\newcommand{\CNOT}{\mathsf{CNOT}}

\newcommand{\tensor}{\otimes}
\newcommand{\bang}{\mathop{!}\nolimits} 

\newcommand{\typetwo}{B}
\newcommand{\typethree}{C}

\newcommand{\bor}{\,|\,}

\newcommand{\PCFpair}[1]{{\langle{#1}\rangle}}

\newcommand{\PCFtt}{{\tt t\!\!t}}
\newcommand{\PCFff}{{\tt f\!\!f}}

\newcommand{\PCFif}[3]{{{\tt if}\,{#1}\,{\tt then}\,{#2}\,{\tt
      else}\,{#3}}}
\newcommand{\PCFletrec}[4]{{{\tt letrec}\,{#1}\,{#2}={#3}\,{\tt in}\,{#4}}}

\newcommand{\PCFarrow}{\multimap}
\newcommand{\PCFprod}{\otimes}
\newcommand{\PCFbang}{{!}}
\newcommand{\PCFletp}[3]{{{\tt let}\,{\langle}{#1}{\rangle}={#2}\,{\tt
      in}\,{#3}}}
\newcommand{\PCFnew}{{\tt new}}
\newcommand{\PCFalpha}{\alpha}

\newcommand{\PCFam}{\textbf{M}}

\newcommand{\PCFentail}{\vdash}

\newcommand{\PCFcbv}{\to}

\newcommand{\PCFreddcbvx}{\mathrel{{\redd}^*}}

\newcommand{\PCFsemcbv}[1]{{#1}^\dagger}
\newcommand{\net}[1]{\PCFsemcbv{#1}}
\newcommand{\PCFAM}{$\mathsf{PCF}_{\mathsf{AM}}$}

\newcommand{\setone}{A}

\newcommand{\elone}{a}

\newcommand{\dstone}{\mu}
\newcommand{\dsttwo}{\rho}
\newcommand{\dists}[1]{\mathit{DST}(#1)}

\newcommand{\supp}[1]{\mathit{SUPP}(#1)}
\newcommand{\zerodst}{\mathbf{0}}
\newcommand{\parsone}{\mathcal{A}}
\newcommand{\redone}{\rightarrow}
\newcommand{\red}{\redone}

\newcommand{\reach}{\mathrel{\looparrowright}}

\newcommand{\redd}{\mathrel{\rightrightarrows}}

\newcommand{\oreach}{\mathrel{{\reach^\circ}}}

\newcommand{\NF}[1]{\mathcal{T}(#1)}

\newcommand{\term}[1]{#1^\circ}
\newcommand{\cont}[1]{\bar #1}

\newcommand{\inputs}[1]{{\tt Inputs}(#1)}
\newcommand{\syncnode}[1]{{\tt SyncNode}(#1)}

\newcommand{\ind}[1]{{\tt ind}_{#1}}
\newcommand{\mapsyncname}[1]{{\tt mkname}_{#1}}
\newcommand{\mem}[1]{\mathbf{m}_{#1}} 
\newcommand{\R}{\mathbf{R}}
\renewcommand{\S}{\mathbf{S}}
\newcommand{\test}{\mathrm{test}}
\newcommand{\upd}{\mathrm{update}}

\newcommand{\mm}{\mathbf{m}}

\newcommand{\prognetsset}{\mathcal{N}}

\newcommand{\ket}[1]{|{#1}\rangle}

\newcommand{\stT}{T}
\newcommand{\stU}{U}

\newcommand{\stTT}{\mathbf{T}}
\newcommand{\stUU}{\mathbf{U}}

\newcommand{\rstT}{\mathcal{T}}
\newcommand{\rstU}{\mathcal{U}}

\newcommand{\stII}{\mathbf{I}}

\newcommand{\stsI}[1]{\sts_{\stII_{#1}}}

\newcommand{\Nat}{\mathbb{N}}

\newcommand{\Finbij}{\rm Perm}
\newcommand{\perm}{\Finbij}

\newcommand{\Msupp}{{\rm supp}}

\newcommand{\Mtest}{{\rm test}}
\newcommand{\Mupdate}{{\rm update}}
\newcommand{\Marity}{{\rm arity}}
\newcommand{\Mswap}{\cdot}
\newcommand{\Mtrue}{{\rm true}}
\newcommand{\Mfalse}{{\rm false}}

\newcommand{\memories}{\mathrm{Mem}} 

\newcommand{\Imem}{\mathcal{I}}
\newcommand{\Pmem}{\mathcal{P}}
\newcommand{\Qmem}{\mathcal{Q}}

\begin{document}
\toappear{}

\setlength{\pdfpageheight}{\paperheight}
\setlength{\pdfpagewidth}{\paperwidth}

\conferenceinfo{CONF 'yy}{Month d--d, 20yy, City, ST, Country}
\copyrightyear{20yy}
\copyrightdata{978-1-nnnn-nnnn-n/yy/mm}
\copyrightdoi{nnnnnnn.nnnnnnn}


\titlebanner{}        
\preprintfooter{short description of paper}   

\title{The Geometry of Parallelism}
\subtitle{Classical, Probabilistic, and Quantum Effects\condinc{}{\\
  \textit{Extended version}}}

\authorinfo{Ugo Dal Lago}
           {Universit\`a di Bologna, Italy \&\\ INRIA, France}
           {ugo.dallago@unibo.it}
\authorinfo{Claudia Faggian}
           {CNRS \& Univ. Paris Diderot, France}
           {faggian@irif.fr}
\authorinfo{Beno\^{i}t Valiron}
           {LRI, CentraleSup\'{e}lec, Univ. Paris-Saclay, France}
           {benoit.valiron@lri.fr}
\authorinfo{Akira Yoshimizu}
           {The University of Tokyo, Japan}
           {yoshimizu@is.s.u-tokyo.ac.jp}

\maketitle

\begin{abstract}
  We introduce a Geometry of Interaction model for higher-order
  quantum computation, and prove its \emph{adequacy} for \emph{a fully
    fledged quantum programming language} in which entanglement,
  duplication, and recursion are all available.   This model is an instance of a new framework  which captures not only quantum but also
  classical and \emph{probabilistic} computation. Its main feature is the
  ability to model \emph{commutative effects} in a \emph{parallel}
  setting. Our model comes with a
  multi-token machine, a proof net system, and a \PCFwo-style
  language. Being based on a multi-token machine equipped with a
  memory, it has a concrete nature which makes it well suited for
  building low-level operational descriptions of higher-order
  languages.
\end{abstract}

\category{F.3.2}{Semantics of Programming Languages}{Algebraic
  approaches to semantics}


\keywords
Geometry of Interaction, memory structure, quantum, probabilistic, PCF

\section{Introduction}
In classical computation, information is deterministic, discrete and
freely duplicable.  Already from the early days~\cite{ShannonShapiro},
however, determinism has been relaxed by allowing state evolution to
be probabilistic. The classical model has then been further challenged
by quantum computation~\cite{NielsenChuang}, a computation paradigm
which is based on the laws of quantum mechanics.

Probabilistic and quantum computation are both justified by the very
efficient algorithms they give rise to: think about Miller-Rabin
primality test~\cite{miller-primality,rabin-primality}, Shor's
factorization~\cite{shor} but also the more recent algorithms for
quantum chemistry~\cite{gse} or for solving linear systems of
equations~\cite{qls}.
 Finding out a way to conveniently
\emph{express} those algorithms without any reference to the
underlying hardware, is then of paramount importance. 

This has stimulated research on programming languages for
probabilistic~\cite{Koz81,Plo82} and quantum
computation (see \cite{GaySurvey} for a survey), and recently on higher-order functional
languages~\cite{SelingerValiron, quipper, PaganiSV14}.  The latter has
been epitomized by variations and extensions of the
$\lambda$-calculus. In order to allow compositional reasoning, it is
important to give a denotational semantics to those languages; maybe
surprisingly, a large body of works in this direction is closely
connected to denotational and interactive models of Linear Logic~\cite{girard87}, in
the style of Game Semantics~\cite{AbramskyJM00,HylandO00} and the
Geometry of Interaction~\cite{Girard89}.

The case of quantum computation is emblematic. The first adequate
denotational model for a quantum programming language \emph{\`a la} \PCFwo, only
two years old~\cite{PaganiSV14}, marries a categorical construction
for the exponentials of linear
logic~\cite{lafont-phd,mellies-panoramas} to a suitable extension of
the standard model of quantum computation: the category of completely
positive maps~\cite{selinger04}. The development of an \emph{interactive}
semantics has proved to be highly nontrivial, with results which are
impressive but not yet completely satisfactory. In particular, the
underlying language either does not properly reflect
entanglement~\cite{hasuo11,Delbecque,DelbecquePanangaden}, a key
feature of quantum computation, or its expressive power is too weak,
lacking recursion and duplication
\cite{DalLagoZorziLINEARITY,LICS2014}. The main reason for this
difficulty lies in the inherent non-locality of
entanglement~\cite{NielsenChuang}.

In this paper we show that Girard's Geometry of Interaction (GoI)
indeed offers the right tools to deal with a fully fledged quantum
programming language in which duplication and full recursion are
available, when we equip GoI with an external quantum memory, a standard
technique for operational models of quantum
$\lambda$-calculi~\cite{SelingerValiron}.

We go further: the approach we develop is not specific to quantum
computation, and our quantum model is introduced as an instance of a new framework
which models \emph{choice effects} in a parametric way, via a {\em
  memory structure}.  The memory structure comes with three
operations: (1) \emph{allocation} of fresh addresses in the memory,
(2)\emph{ low-level actions} on the memory, and (3) \emph{choice}
based on the value of the memory at a given address. The notion of
memory structure is flexible, the only requirement being
commutativity of the operations. In Sec.~\ref{sec:exemples} we show
that different kinds of choice effects can be systematically treated:
classical, probabilistic and quantum memory are all instances of this
general notion. Therefore, the memory makes the model suitable to
interpret classical, probabilistic and quantum functional programs. In particular, in 
the case of quantum memory, a low-level action is an application of
unitary gate to the  memory, while the choice performs a
quantum measure.

The GoI model we give has a very concrete nature, as it consists of a
class of token machines~\cite{DanosRegnier, Mackie95}.  Their
distinctive feature is to be {\em parallel} and
\emph{multi-token}~\cite{LICS2014,lics2015} rather than
\emph{single-token} as in classic token machines~\cite{DanosRegnier}.
Being multi-token means that different computational threads can
interact with each other and synchronize (think of this as a
multi-player game, where players are able to collaborate and exchange
information). The presence of multiple tokens allows to appropriately
reflect non-locality in a quantum setting, but also to generally  deal with
\emph{parallelism} and \emph{choice effects} in a satisfactory way. We
discuss why this is the case when we concretely present the machine
(Sec. \ref{PSIAM}).

Finally, to deal with the combination of parallelism, probabilistic
side-effects and non-termination, we  develop a general notion of
PARS, {\em probabilistic abstract rewrite system}. The results we
establish on PARS are the key ingredient in the Adequacy proofs, but are  also of independent interest. The
issues at sake are non-trivial and we discuss them in the dedicated
Sec.~\ref{sec:par-cong-term}.
\paragraph*{Contributions.}
We present a Geometry of Interaction (GoI) model for higher-order quantum
computation, which is adequate for a quantum programming language in
which entanglement, duplication, and recursion are all available. Our
model comes with a multi-token machine, a proof net system, and a
\PCFwo-style language.
%
More specifically, this paper's contributions can be summarized
as follows:
\begin{varitemize}
\item 
  we equip GoI with the ability to capture
  \emph{choice effects} using a parametric notion of memory
  structure (Sec.~\ref{sect:alg-mem});
\item 
  we show that the notion of memory structure is able to capture
  classical, probabilistic and quantum effects (Sec.~\ref{sec:exemples});
\item 
  we introduce a construction which is \emph{parametric} on the
    memory, and produces a class of multi-token machines
  (Sec.~\ref{PSIAM}), proof net systems (Sec.~\ref{sec:prog-net}) and
  \PCFwo-style languages (Sec.~\ref{sect:PCF}). We prove that
  (regardless of the specific memory) the multi-token machine is an
  adequate model of nets reduction (Th.~\ref{adequacy psiam}), and the nets an adequate model of
  \PCFwo\ term rewriting (Th.~\ref{th:adeqcbv});
\item 
  we develop a general notion of parallel abstract rewrite system,
  which allows us to deal with the combination of parallelism and
  probabilistic choice in an infinitary setting
  (Sec.~\ref{pars}).
\end{varitemize}
Being based on a multi-token machine associated to a memory, our model
has a concrete nature which makes it well suited to build low-level
operational descriptions of higher-order programming languages. In the remainder
of this section, we give an informal overview of various aspects of
our framework, and motivate with some examples the significance of our
contribution. 

\condinc{
An extended version of this paper with proofs and more
details is available online \cite{LVarxiv}.
}{This report is an extended version of \cite{popl17}.}
%
%

\subsection{Geometry of Interaction and Quantum Computation}
Geometry of Interaction is interesting as semantics for programming
languages \cite{phdmackie,Mackie95,GhicaSS11} because it is a
high-level semantics which at the same time is close to low-level
implementation and has a clear operational flavor.  Computation is
interpreted as a flow of information circulating around a network,
which essentially is a representation of the underlying program.
Computational steps are broken into low-level actions of one or more
tokens which are the agents carrying the information around.  A long
standing open question is whether fully fledged higher-order quantum
computation can be modeled operationally via the Geometry of Interaction.

\subsubsection{Quantum Computation}
\label{sec:quantum-computation}
As comprehensive references can be found in the
literature~\cite{NielsenChuang}, we only cover the very few concepts
that will be needed in this paper. Quantum computation deals with {\em
  quantum bits} rather than bits.
The state of a quantum system can be represented with a density matrix
to account for its probabilistic nature. However for our purpose we
shall use in this paper the usual, more operational, non-probabilistic
representation.
Single quantum bits (or {\em qubits}) will thus be represented by a
ray in a two-dimensional complex Hilbert space, that is, an
equivalence class of non-zero vectors up to (complex) scalar
multiplication. Information is attached to a qubit by choosing an
orthonormal basis $(\ket{0},\ket{1})$: a qubit is a superposition of
two classical bits (modulo scalar multiplication). If the state of
several bits is represented with the product of the states of single
bits, the state of a multi-qubit system is represented with the {\em
  tensor product} of single-qubit states. In particular, the state of
an $n$-qubit system is a superposition of the state of an $n$-bit
system.
We consider superpositions to be normalized.

Two kinds of operations can be performed on qubits. First, one can
perform reversible, \emph{unitary gates}: they are unitary maps in the
corresponding Hilbert space. A more exotic operation is the
\emph{measurement}, which is the only way to retrieve a classical bit out of a
quantum bit. This operation is probabilistic: the probabilities depend
on the state of the system. Moreover, it modifies the state of the
memory. Concretely, if the original memory state is
$\alpha_0\ket0\otimes\phi_0 + \alpha_1\ket1\otimes\phi_1$ (with
$\phi_0$ and $\phi_1$ normalized), measuring the first qubit would
answer $x$ with probability $|\alpha_x|^2$, and the memory is turned
into $\ket{x}\otimes\phi_x$. Note how the measurement not only
modifies the measured qubit, but also collapses the global state of the
memory.

The effects of measurements are counterintuitive especially in {\em
  entangled} system: consider the 2-qubit system
$\frac{\sqrt2}2(\ket{00} + \ket{11})$. This system is entangled,
meaning that it cannot be written as $\phi\otimes\psi$ with $1$-qubit
states $\phi$ and $\psi$. One can get such a system from the state
$\ket{00}$ by applying first an {\em Hadamard gate} $\Had$ on the
second qubit, sending $\ket{0}$ to $\frac{\sqrt2}2(\ket0+\ket1)$ and
$\ket1$ to $\frac{\sqrt2}2(\ket0-\ket1)$, therefore getting the state
$\frac{\sqrt2}2(\ket{00} + \ket{01})$, and then a $\CNOT$ ({\em
  controlled-not}) gate, sending $\ket{xy}$ to
$\ket{x\oplus{}y}\otimes\ket{y}$. Measuring the first qubit will
collapse the entire system to $\ket{00}$ or $\ket{11}$, with equal
probability $\frac12$.

\begin{remark}\label{rem:nonlocal}
  Notwithstanding the global collapse induced by the measurement, the
  operations on physically disjoint quantum states are
  commutative. Let $A$ and $B$ be two quantum states. Let $U$ act on
  $A$ and $V$ act on $B$ (whether they are unitaries, measurements, or
  combinations thereof). Consider now $A\otimes B$: applying $U$ on
  $A$ then $V$ on $B$ is equivalent to first applying $V$ on $B$ and
  then $A$ on $U$. In other words, the order of actions on physically
  separated quantum systems is irrelevant. We use this property in
  Sec.~\ref{sec:quantum-memory}.
\end{remark}

\subsubsection{Previous Attempts and Missing Features}
\label{sec:previous-attempts}\label{related_quantum}
A first proposal of Geometry of Interaction for quantum computation is
\cite{hasuo11}. Based on a purely categorical
construction~\cite{Abramsky}, it features duplication but {not}
general entanglement: entangled qubits cannot be separately acted
upon. As the authors recognize, a limit of their approach is that
their GoI is single-token, and they already suggest that using several
tokens could be the solution.

\begin{example}\label{ex:entangled}
As an example, if $S=\frac{\sqrt2}2(\ket{00} + \ket{11})$, the term
\begin{equation} \label{eq:entangled}
  {\tt let~} x\otimes y = S{\tt~in~}(U x) \otimes (V y)
\end{equation}
cannot be represented in \cite{hasuo11}, because it is not
possible to send entangled qubits to separate parts of the
program.
\end{example}
A more recent proposal \cite{LICS2014}, which introduces an
operational semantics based on {\em multi-tokens}, can handle general
entanglement.  However, it does neither handle duplication nor
recursion. More than that, the approach relies on termination to
establish its results, which therefore do not extend to an infinitary
setting: it is not enough to ``simply add'' duplication and fix points.

\begin{example}\label{ex:recursion}
In \cite{LICS2014} it is not possible to simulate the program that
tosses a coin (by performing a measurement), returns a fresh qubit on
head and repeats on tail. In mock-up \ML, this program becomes
\[
  \PCFletrec{f}{x}{(\PCFif{x}{\PCFnew}{f\,(\Had\,\PCFnew}))}{(f\,(\Had\,\PCFnew))}
\]
where {\tt new} creates a fresh qubit in state $\ket0$ and where the
{\tt if} test performs a measurement on the qubit $x$. Note how the
measure of $\Had\,{\PCFnew}$ amounts to tossing a fair coin:
$\Had\,{\PCFnew}$ produces $\frac{\sqrt2}2(\ket0{+}\ket1)$.  Measuring
gives $\ket0$ and $\ket1$ with probability $\frac12$.
\end{example}
Example~\ref{ex:recursion} will be our leading example all along the
paper. Furthermore, we shall come back to both examples in
Sec.~\ref{sec:discussion}.

\subsection{Parallel Choices: Confluence and Termination}
\label{sec:par-cong-term}\label{sec:intro_conf}
When dealing with both probabilistic choice and infinitary
reduction, parallelism makes the study of confluence and convergence
highly non-trivial.
The issue of confluence arises as soon as choices and duplication are
both available, and non-termination adds to the challenges. Indeed, it
is easy to see how tossing a coin and duplicating the result does not
yield the same probabilistic result as tossing twice the coin.
To play  with this, let us take for  example  the following  term of the probabilistic
$\lambda$-calculus  \cite{EhrhardTassonPagani}:~ 
$
M=(\lambda x.x\;\mathtt{xor}\;x)((\mathtt{tt}\oplus\mathtt{ff})\oplus\Omega)
$
where $\mathtt{tt}$ and $\mathtt{ff}$ are boolean constants, $\Omega$
is a divergent term, $\oplus$ is the choice operator (here, tossing a
fair coin), and $\mathtt{xor}$ is the boolean operator computing the
exclusive or. Depending on which of the two redexes we fire first, $M$
will evaluate to either the distribution
$\{\mathtt{ff}^{\frac{1}{2}}\}$ or to the distribution
$\{\mathtt{tt}^{\frac{1}{8}},\mathtt{ff}^{\frac{1}{8}}\}$.  In
ordinary, deterministic \PCFwo, any program of boolean type may or may
not terminate, depending on the reduction strategy, but its normal
form, if it exists, is unique.  This is not the case for our 
probabilistic term $M$: depending on the choice of the redex, it
evaluates to two distributions which are simply not comparable.

In the case of probabilistic $\lambda$-calculi, the way-out to this  is to fix
a reduction strategy; the issue however is not only in the syntax, it
appears---at a more fundamental  level---also in the model. This is the case for
\cite{EhrhardTassonPagani}, where the model itself does not support
parallel probabilistic choice.  Similarly, in the development of a
Game Semantics or Geometry of Interaction model for probabilistic
$\lambda$-calculi, the standard approach has been to use a polarized
setting, so to impose a strict form of
sequentiality~\cite{danosharmer,HoshinoMH14}.
If instead we  choose to have parallelism in  \emph{the model}, confluence is
not granted and even the \emph{definition} of convergence is
non-trivial.

In this paper we propose   a probabilistic  model that is
{\em infinitary and parallel but confluent}; to achieve this, in Sec.~\ref{pars} we develop  some 
results which are general to any probabilistic abstract rewrite system, and which to our knowledge are  novel.
 More specifically, we
provide sufficient conditions for  an infinitary probabilistic system to be confluent and to
satisfy a property which is a probabilistic analogous of the familiar
``weak normalization implies strong normalization''. 
We then show that the parametric
models which  we introduce (both the proof nets and the multi-token machine) satisfy this property; this  is indeed what  ultimately  grants  the 
adequacy results.

\subsection{Overview of the Framework, and Its Advantages}\label{overview}
A quantum program has on one hand features which are specific to quantum computing, and on the other hand standard constructs.  This is indeed 
the case for many paradigmatic languages; analyzing the features
separately is often very useful.  Our framework clearly separates
(both in the language and in its operational model) the constructs
which are common to all programming languages (e.g. recursion)
and the features which are specific to some of them (e.g. measurement or  probabilistic choice). The former is captured by a {\em fixed operational core}, the
latter is encapsulated within a {\em memory structure}. This  
 approach   has two
distinctive advantages:
\begin{varitemize}
\item 
  \emph{Facilitate Comparison between Different Languages}: clearly
  separating in the semantics the standard features from the ``notions
  of computation'' which is the specificity of the language, allows for
  an easier comparison between different languages.
\item 
  \emph{Simplify the Design of a Language with Its Operational Model}:
  it is enough to focus on the memory structure which encapsulates the
  desired effects. Once such a memory structure is given, the construction
  provides an adequate Geometry of Interaction model for a \PCFwo-like
  language equipped with that memory.
\end{varitemize}
%
Memory structures are defined in Sec.~\ref{sect:alg-mem}, while the 
operational core is based on Linear Logic: a
\emph{linearly-typed} \PCFwo-like language, 
\emph{Geometry of Interaction}, and its syntactical counterpart,
\emph{proof nets}. Proof nets are a graph-based formal system that
provides a powerful tool to analyze the execution of terms as a
rewriting process which is mostly parallel, local, and asynchronous.

More in detail, our framework consists  of:
\begin{varenumerate}
\item
  a notion of\emph{ memory structure}, whose operations are suitable
  to capture a range of choice effects;
\item
  an operational core, which is articulated in the \emph{three {base}
    rewrite systems} (a proof net system, a GoI multi-token machine,
  and a \PCFwo-style language);
\item
  a construction which is \emph{parametric on the memory}, and lifts
  each {base} rewrite system into a more expressive operational
  system. We respectively call these systems: program nets, \PSIAM\ machines and
  \PCFAM{} abstract machines.
\end{varenumerate}
Finally, the three forms of
systems  are \emph{all related by adequacy results}.
As long as the memory operations satisfy commutativity, the
construction produces an adequate GoI model for the corresponding \PCFwo\
language.  More precisely, we prove--- again {\em parametrically on the
memory}---that the \PSIAM\ is an adequate model of program net
reduction (Th.~\ref{adequacy psiam}), and program nets are
expressive enough to adequately represent the behavior of the \PCFwo{}
language (Th.~\ref{th:adeqcbv}).

\subsection{Related Work}
The low-level layer of our framework can be seen as a generalization
and a variation of systems which are in the literature. The nets and
multi-token machine we use are a variation of \cite{lics2015}, the
linearly typed \PCFwo\ language is the one in~\cite{PaganiSV14} (minus
lists and coproducts).  What we add in this paper are the right tools
to deal with challenges like probabilistic parallel reduction and
entanglement. Neither quantum nor probabilistic choice can be treated
in \cite{lics2015}, because of the issues we clarified in Sec.
\ref{sec:intro_conf}.  The specificity of our proposal is
really its ability to deal with \emph{choice} together with
\emph{parallelism}.

We already discussed previous attempts to give a GoI model of quantum
computation, and their limits, in Sec. \ref{related_quantum}
above. Let us quickly go through other models of quantum
computation. Our parametric memory is presented equationally:
equational presentations of quantum memory are common in the
literature~\cite{staton,oxford}. Other models of quantum memories are
instead based on Hilbert spaces and completely positive maps, as
in~\cite{PaganiSV14,selinger04}. In both of these approaches, the
model captures with precision the structure and behavior of the
memory. Instead, in our setting, we only consider the interaction
between the memory and the underlying computation by a set of
equations on the state of the memory at a given address, the
allocation of fresh addresses, and the low-level actions.

Finally, taking a more general perspective, our proposal is by no
means the first one to study effects in an interactive setting.
Dynamic semantics such as GoI and Game Semantics are gaining interest
and attention as semantics for programming languages because of their
operational flavor.  \cite{Ghica07,hasuomuroya,HoshinoMH14} all deal
with effects in GoI.  A common point to all these works is to be
single-token.  While our approach at the moment only deals with choice
effects, we indeed deal with \emph{parallelism}, a challenging feature
which was still missing.

\section{PARS: Probabilistic Abstract Reduction Systems}\label{sect:confluence}\label{pars}\label{sect:prelim}

Parallelism allows critical pairs; as we observed in
Sec. \ref{sec:intro_conf}, firing different redexes will produce different
distributions and can lead to possibly very different results. Our parallel
model however enjoys a property similar to the diamond property
of abstract rewrite systems. Such a property entails a number of
important consequences for confluence and normalization,
and these results in fact are general to any probabilistic abstract
reduction system. 
In particular, we define what we mean by strong and weak normalization
in a probabilistic setting, and we prove that \emph{a suitable
  adaptation} of the diamond property guarantees confluence and a form
of \emph{uniqueness of normal forms}, not unlike what happens in the
deterministic case.  Th.~\ref{th:proba_term} is the main result of
the section.

In a probabilistic context, spelling out the diamond property requires
some care. We will introduce a strongly controlled notion of reduction
on distributions ($\redd$). The need for this control has the same
roots as in the deterministic case: please recall that strong
normalization follows from weak normalization by the diamond property
%
(~$b \leftarrow a \rightarrow c \Rightarrow b=c \lor \exists d (
b \rightarrow d \leftarrow c ) $~) but \emph{not} from
subcommutativity (~$ b \leftarrow a \rightarrow c \Rightarrow \exists
d ( b \rightarrow^= d \leftarrow^= c $)~) which appears very similar,
but ``leaves space'' for an infinite branch.

\subsection{Distributions and PARS}\label{sect:distrpars}
We start by  setting the  basic definitions.
Given a set $\setone$, we note $\dists{\setone}$ for the set
of \emph{probability distributions} on $\setone$: any
$\dstone\in\dists{\setone}$ is a function from $\setone$ to $[0,1]$
such that $\sum_{\elone\in\setone}\dstone(\elone)\leq 1$. A distribution $\dstone$ is \emph{proper} if
$\sum_{\elone\in\setone}\dstone(\elone)=1$. 
\condinc{}{The distribution assigning
$0$ to each element of a set $\setone$ is indicated with
$\zerodst$. } We indicate with $\supp{\dstone}$
the support of a distribution $\dstone$, \ie\ the subset of
$\setone$ whose image under $\mu$ is not $0$.
On  $\dists{\setone}$, we  define the relation $\subseteq$  point-wise:  $\dstone \subseteq \dsttwo$ if
$\dstone (a) \leq \dsttwo (a)$ for each $a\in A$.

A \emph{probabilistic abstract reduction system (PARS)} is a pair
$\parsone=(\setone,\redone)$ consisting of a set $\setone$ and a
relation ${\redone}\subseteq{\setone\times\dists{\setone}}$ (rewrite
relation, or reduction relation) such that for each $(a,\dstone)\in{\redone}$,
$\supp{\dstone}$ is finite.  We write $a \redone \mu$ for
$(a,\mu)\in{\redone}$. An element $a\in \setone$ is \emph{terminal} or
in \emph{normal form} (w.r.t. $\redone$) if there is no $\mu$ with
$a\redone \mu$, which we write $a \not \redone$.

We can partition any distribution $\mu$ into a distribution
$\mu^\circ$ on terminal elements, and a distribution $\bar\mu$ on
elements for which there exists a reduction, as follows:
$$
\mu^\circ(a)=
  \left\{
    \begin{array}{ll}
       \mu(a)&\mbox{if $a\not\redone$,}\\  
       0&\mbox{otherwise;}
    \end{array}
  \right.
\qquad\qquad
\bar\mu(a)=\mu(a)-\mu^\circ(a).
$$
The \emph{degree of termination} of 
$\dstone$, written $\NF{\dstone}$, is 
$\sum_{a\in\setone}\dstone^\circ(a)$.


We write $\reach$ for the \emph{reflexive and transitive closure} of
$\redone$, namely the smallest subset of
$\setone\times\dists{\setone}$ closed under the following rules:
\condinc{
if
$a\redone\mu$ then $a\reach \mu$; we always have $a \reach \{a^1\}$;
whenever $a\reach \mu + \{b^p\}$, $b\reach \rho$ and
$b\notin \supp{\mu}$ we have $a\reach \mu + p\cdot \rho$. }
{$$
\infer{a\reach\mu}{a\redone \mu}\qquad 
\infer{a \reach \{a^1\}}{}  \qquad 
\infer{a\reach \mu + p\cdot \rho}{a\reach \mu + \{b^p\} & b\reach \rho & b\notin \supp{\mu}}
$$}
 We read
$a\reach \mu$ as \emph{``$a$ reaches $\mu$''}.


\noindent\textbf{The Relation $\pmb{\redd}$.}
In order to extend to PARS classical results on termination for
rewriting systems, we  define the  binary relation
$\redd$, which lifts the notion of one
step reduction to distributions: we require that \emph{all} non-terminal elements are indeed
reduced. The relation
$ {\redd} \subseteq \dists\setone\times\dists{\setone}$ is defined as
$$          
\infer
    {\mu \redd \term{\mu} + \sum_{ a\in\supp{\cont{\dstone}}}\dstone(a)\cdot\dsttwo_a}
    {\mu =  \term{\mu} +  \cont{\mu}  &   \{a \redone \dsttwo_a\}_{ a\in\supp{\cont{\dstone}}} }\,.
    $$
    Please note that in the derivation above, we require
    $a \redone \dsttwo_a$ \emph{for each} $a\in\supp{\cont{\dstone}}$.
    Observe also that $\term{\mu} \redd \term{\mu}$
    since $\supp{\cont{{\term{\mu}}}} = \emptyset$.

    We write $\mu \redd^n \rho$ if $\mu$ reduces to $\rho$ in $n\geq 0$
steps; we write $\mu \redd^* \rho$ if there is any \emph{finite}
sequence of reductions from $\mu$ to $\rho$.


With a slight abuse of notation, in the rest of the paper we sometime
write  $\{a\}$   for $\{a^1\}$,  or simply  $a$ when clear from the context. 
As an example, we  write $a \redd \mu$ for $\{a^1\} \redd \mu$.
Moreover, the distribution $\{a_1^{p_1},\ldots,a_n^{p_n}\}$ will be
often indicated as $\sum p_i\cdot\{a_i\}$ thus facilitating algebraic
manipulations.

\subsection{Normalization and Confluence}
In this subsection, we look at  normalization and confluence
in the probabilistic setting, which we introduced  in Sec. \ref{sect:distrpars}.
We need to distinguish between \emph{weak} and \emph{strong} normalization.
The former refers to the \emph{possibility} to reach normal forms following
any reduction order, while the latter (also known as \emph{termination}, see~\cite{ARS}) refers to the \emph{necessity}
of reaching normal forms. In both cases, the concept is inherently quantitative.
\begin{definition}[Weak and Strong Normalization]\label{def:normalization}
Let $p\in [0,1]$ and let $\mu\in\dists{\setone}$ Then:
\begin{varitemize}
\item 
   $\mu$ \emph{weakly $p$-normalizes} (or weakly normalizes with probability at least $p$) if there exists $\rho$ such
   that $\mu \redd^* \rho$ and $\NF{\rho}\geq p$.
\item 
  $\mu$ \emph{strongly $p$-normalizes}
   (or strongly normalizes with probability at least $p$) 
    if there exists $n$ such that
$\mu \redd^n \rho$ implies $\NF{\rho}\geq p$, for all $\rho$. 
\end{varitemize}
The relation $\redone$ is said \emph{uniform} if for each $p$, and
each $\mu \in \dists{\setone}$, weak $p$-normalization implies strong
$p$-normalization. 

\condinc{}{Following ~\cite{ARS}, we will also use the term  \emph{$p$-termination} for  \emph{strong}  $p$-normalization, and  refer to weak $p$-normalization as simply $p$-normalization.}
\end{definition}
Even the mere notion of convergent computation must be made
quantitative here:
\begin{definition}[Convergence]\label{def:limit} 
The distribution $\mu\in\dists{\setone}$  \emph{converges with probability $p$}, written $\mu \Downarrow_p$, if
 $p = {\sup_{\mu \redd^* \rho}\NF{\rho}}.$
\end{definition}
Observe that for every $\mu$ there is a unique probability $p$ such that
$\mu\Downarrow_p$.  Please also observe how Definition \ref{def:limit} is taken
over all $\rho$ such that $\mu \redd^* \rho$, thus being forced to
take into account all possible reduction orders. If $\red$ is uniform,
however, we can reach the limit along \emph{any} reduction order:
\begin{proposition}\label{prop:equiv}
Assume $\red$ is uniform. Then for every $\mu$ such that
$\mu \Downarrow_p$ and for every sequence of distributions ${(\rho_n)_n}$
such that $\mu= \rho_0$ and $\rho_{n}\redd \rho_{n+1}$ for every $n$,
it holds that $p=\lim_{n\to \infty} \NF{\rho_n}$.
\qed
\end{proposition}
\condinc{}{\begin{remark}
  Observe that because of Prop. \ref{prop:equiv}, ${\sup_{\mu \redd^*
      \xi}\NF{\rho}}={\sup_{\mu \reach \rho}\NF{\rho}}$.
\end{remark}}
A PARS is said to be confluent iff $\redd$ is a confluent relation
in the usual sense:
\begin{definition}[Confluence]
The PARS $(\setone,\redone)$ is said to be \emph{confluent} if whenever
$\tau \redd^* \mu$ and $\tau \redd^* \xi$, there exists $\rho$ such
that $\mu\redd^* \rho$ and $\xi\redd^* \rho$.
\end{definition}

\subsection{The Diamond Property in a Probabilistic Scenario}
In this section we study a property which guarantees confluence and
uniformity.

\begin{definition}[Diamond Property for PARS]\label{base}
   A PARS $(A,\redone)$  \emph{satisfies the diamond
  property} if the following holds.
   Assume  $\mu\redd \nu$
  and $\mu \redd \xi$. Then~~(1)
  $\term{\nu} = \term{\xi} $
  and ~~(2) $\exists \rho$ with  $\nu\redd\rho$ and $\xi\redd\rho$.
\end{definition}

As an immediate consequence:
\begin{corollary}[Confluence]
If $(A,\red)$  satisfies  the diamond property, then $(A,\red)$ is
confluent. \qed
\end{corollary}
Finally, then, the diamond property ensures that weak $p$-normaliza\-tion
implies strong $p$-normalization, precisely like for usual abstract rewrite systems:
\begin{theorem}[Normalization  and Uniqueness of Normal Forms]\label{th:proba_term}
  Assume  $(A,\red)$  satisfies the diamond
  property. Then:
\begin{varenumerate}
\item
  \textbf{Uniqueness of normal forms.}  $ \mu \redd^k \rho$ and $\mu
  \redd^k \tau$ for some $k \in \Nat$ implies $\rho^\circ = \tau^\circ$.
\item
  \textbf{Uniformity}. If $\mu$ is weakly $p$-normalizing (for some
  $p\in [0,1]$), then $\mu$ strongly $p$-normalizes, \ie, $\red$ is
  uniform.
\end{varenumerate}
\end{theorem}
\begin{proof}
First note that (2) follows from (1). In order to prove (1), we use
  an adaptation of the familiar ``tiling''
argument. It is not exactly the standard proof because reaching some
normal forms in a distribution is not the end of a sequence of
reductions.  Assume $ \mu=\rho_0 \redd \rho_1 \redd ... \redd \rho_k
$, and $ \mu \redd \tau_1 \redd ... \redd\tau_k $.
We prove $\term{\rho_k} = \term {\tau_k} $ by induction on $k$.
If $k=1$, the claim is true by Definition~\ref{base}~(1).
If  $k>1$ we tile (w.r.t. $\redd$), as depicted below:\\[.5ex]
\begin{minipage}{.36\columnwidth}
{ \includegraphics[scale=0.7]{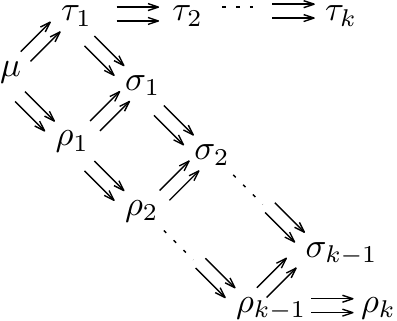}  }
\end{minipage}
 \begin{minipage}{.63\columnwidth}
we build the sequence $\sigma_0=\tau_1$
$\redd \sigma_1 ...\redd \sigma_{k-1}$ (see the Fig. on the side)
where each $\sigma_{i+1}$ ($i\geq 0$) is obtained via Definition~\ref{base}~(2), from $\rho_i \redd \rho_{i+1}$ and $\rho_i \redd \sigma_{i} $, by closing the diamond. 
By  Definition~\ref{base}~(1) $\term{\rho_k}= \term{\sigma_{k-1}}$. 
Now we observe that $\tau_1\redd^{k-1} \tau_k$ and  $\tau_1\redd^{k-1}\sigma_{k-1}$. 
Therefore we have (by induction) 
$ \term{\sigma_{k-1}} =  \term {\tau_k} $, from which we conclude $\term{\rho_k} = \term {\tau_k} $.\qedhere
\end{minipage}  
\end{proof}
 
\section{Memory Structures}\label{sect:alg-mem}

In this section we introduce the notion of memory
structure. Commutativity of the memory operations is ensured by a set
of equations.  To deal with the notion of fresh addresses, and avoid unnecessary bureaucracy,  it is convenient to rely on
nominal sets.  The basic definitions are recalled below (for
details, see, \eg,~\cite{pitts-book}).

\subsection{Nominal Sets}
If $G$ is a group, then a \emph{$G$-set} $(M,\Mswap)$ is a set $M$ equipped
with an action of $G$ on $M$, \ie\ a binary operation
$(\Mswap) : G \times M \longrightarrow M$ which respects the group
operation.
Let $I$ be a countably infinite set; let $M$ be a set equipped with an
action of the group $\Finbij(I)$ of \textit{finitary permutations} of
$I$.
A \emph{support} for $m\in M$ is a subset
$A\subseteq I$ such that for all $\sigma\in\Finbij(I)$,
$\forall i\in A, \sigma i = i$ implies $\sigma\Mswap m = m$. 
A \emph{nominal set} is a $\Finbij(I)$-set all of whose elements have
finite support. In this case, if $m\in M$, we write $\Msupp(m)$ for
the smallest support of $m$.
The complementary notion of support is freshness: $i\in I$ is
\emph{fresh} for $m\in M$ if $i\not\in \Msupp(m)$.
We write $(i~j)$ for the transposition which swaps  $i$ and $j$.

We will make use of the following characterization of support in terms
of transpositions:
$A\subseteq I$ supports $m\in M$ if and only if for every $i , j \in I
- A$ it holds that $(i~ j)\cdot m = m $.  As a consequence, for all
$i,j\in I$, if they are fresh for $m\in M$ then $(i~ j)\cdot m = m$.
\subsection{Memory Structures}\label{sec:mem-struct}
A \emph{memory structure} $\memories =(\memories, I, \Mswap, \mathcal L)$
consists of an infinite, countable set $I$ whose elements
$i,j,k,\dots$ we call \emph{addresses}, a nominal set $(\memories,\Mswap)$
each of whose elements we call \emph{memory states}, or more shortly,
\emph{memories}, and a finite set $\mathcal L$ of \emph{operations}.

We write  $I^*$ for  the set of all tuples made from
  elements of $I$. A tuple is denoted with $(i_1, \ldots, i_n)$, or with
  $\vec i$.
To a memory structure are associated the following maps.
\begin{varitemize}
\item
  $\Mtest : I \times \memories \to \dists{Bool \times \memories}$
  ~~(Observe that the set $\memories$ might be updated by the operation $\Mtest$: for this reason,  it
  also  appears in the codomain -- See Remark~\ref{rem:modif}),
\item $\Mupdate : I^* \times \mathcal L \times \memories \rightharpoonup \memories$ (partial map),
\item $\Marity : \mathcal L \to \Nat$,
\end{varitemize}
and the following three  properties.

\smallskip
\noindent
{\em (1)} The maps $\Mtest$ and $\Mupdate$ respect the
  group action:
  \begin{varitemize}
  \item $\sigma\Mswap(\Mtest(i,m)) = \Mtest(\sigma(i), \sigma\Mswap
    m)$,
  \item $\sigma\Mswap(\Mupdate(\vec{i},x,m)) =
    \Mupdate(\sigma(\vec{i}),x,\sigma\Mswap m)$,
  \end{varitemize}
  where the action of $\perm(I)$ is  extended in the natural way to 
  distributions and pairing with booleans.

\smallskip
\noindent{\em(2)} The action of a given operation on the memory is only
  defined for the correct arity. More precisely, 
  $
    \Mupdate((i_1\ldots i_n),x,m) 
  $
  is defined if and only if
  the $i_k$'s are pairwise disjoint
  and
  $\Marity(x) = n$.

\smallskip
\noindent{\em(3)} Disjoint tests and updates commute: assume
  that $i\neq j$, that $j$ does not meet $\vec{k}$, and that
  $\vec{k}$ and $\vec{k'}$ are disjoint. First,
  updates on $\vec k$ and $\vec k'$ commute:
  $\Mupdate(\vec k, x, \Mupdate(\vec k', x', m)) =
  \Mupdate(\vec k', x', \Mupdate(\vec k, x, m))$.
  Then, tests on $i$ commute with tests on $j$ and
  tests of $j$ commute with updates on $\vec k$. We pictorially
  represent these equations in Fig.~\ref{fig:commut}. The drawings
  are meant to be read from top to bottom and represent the successive
  memories along action. Probabilistic behavior is represented with
  two exiting arrows, annotated with their respective probability of
  occurrence, and the boolean resulting from the test
  operation. Intermediate memories are unnamed and represented with
  ``$\cdot$''.  \condinc{}{We write the formal equations in
    Appendix~\ref{app:commut-mem}.}

\begin{figure}
  \centering
  {\includegraphics[width=\columnwidth]{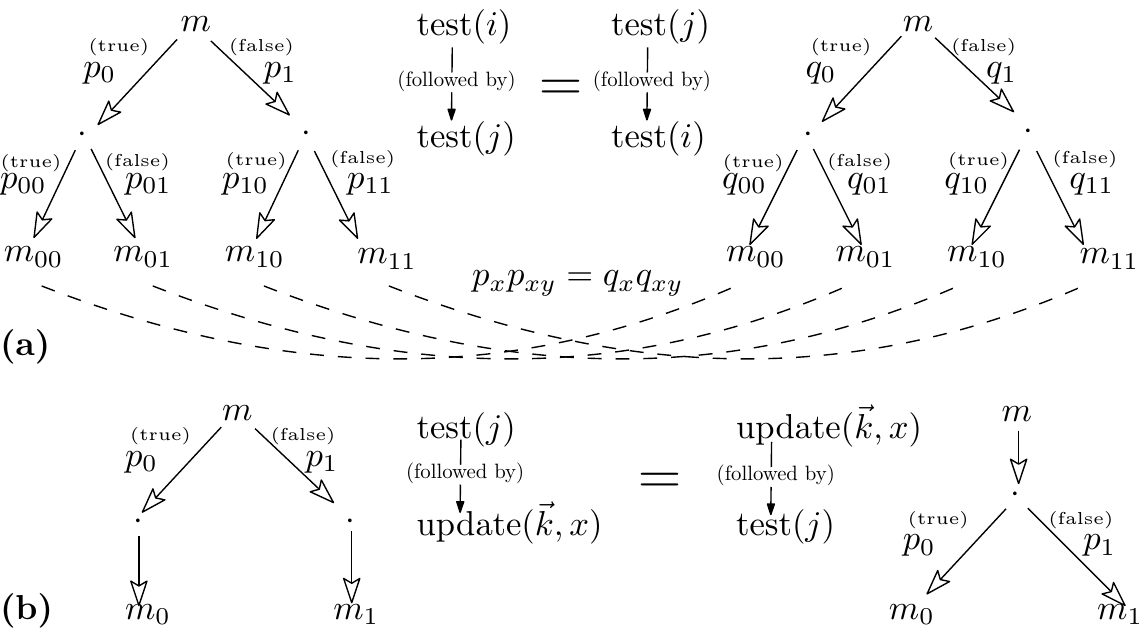}}
  \caption{Commutation of Tests and Updates.}
  \label{fig:commut}
\end{figure}
\subsection{Instances of Memory Structures}\label{sec:exemples}
The structure of memory is flexible and can accommodate several choice
effects. Let us give some relevant examples.  Typically, here $I$ is
$\Nat$, but any countable set would do the job.
\subsubsection{Deterministic, Integer Registers}\label{sec:determ-integ-regist}
The simplest instance of memory structure is the deterministic one,
corresponding to the case of classical \PCFwo\ (this subsumes, in
particular, the case studied in \cite{lics2015}).  Memories are simply
functions $\mm$ from $I$ to $\Nat$, of value 0 apart for a finite
subset of $I$. The test on address $i$ is deterministic, and tests
whether $\mm(i)$ is zero or not.  Operations in $\mathcal L$ may
include the unary predecessor and successor, and for example the
binary max operator.

\begin{example}\label{ex:int-reg}
  A typical representation of this deterministic memory  is a
  sequence of integers:
  indexes correspond to addresses and
  coefficients to values. A completely free memory is for example the
  sequence $\mm_0 = (0,0,0,\ldots)$. If $S$ corresponds to the
  successor and $P$ to the predecessor, here is what happens to the
  memory for some operations. The memory $\mm_1 :=
  \Mupdate(0,S,\mm_0)$ is $(1,0,0,0,\ldots)$, the memory 
  $\mm_2 :=\Mupdate(1,S,\mm_1)$ is $(1,1,0,0,\ldots)$, and
  the memory $\mm_3:=\Mupdate(0,P,\mm_2)$ is $(0,1,0,0,\ldots)$.
  Finally, $\Mtest(1,\mm_3) = ({\rm false},(0,1,0,0,\ldots))$.
    Note that we do not need to keep track of an infinite sequence: a
  dynamic, finite list of values would be enough. We'll come back to
  this in Sec.~\ref{sec:quantum-memory}.
\end{example}

\begin{remark}\label{rem:int-reg}
  The equations on memory structures enforce the fact that all fresh
  addresses (\ie, not on the support of the nominal set)
  have equal values.  Note that
  however the conditions do not impose any particular ``default'' value.
  These equations also state that, in the deterministic case,
  a test action on $i$ can only modify the memory at address
  $i$. Otherwise, it could for example break the commutativity of update and
  test (unless $\mathcal L$ contains trivial operations, only).
\end{remark}

\subsubsection{Probabilistic, Boolean Registers}
\label{sec:prob-bool}

When the test operator is allowed to have a genuinely probabilistic behavior, the
memory model supports the representation of probabilistic boolean
registers. In this case, a memory $\mm$ is a function from $I$
to the real interval $[0,1]$, whose values represent probabilities of
observing ``true''. The test on address $i$ could return
\[
\mm(i)\{({\rm true},\mm\{i\mapsto 1\})\}
+
(1-\mm(i))\{({\rm false},\mm\{i\mapsto 0\})\}
\]
Operations in $\mathcal L$ may for example include a unary ``coin flipping''
operation setting the value associated to $i$ to some fixed
probability.

\begin{example}\label{ex:prob-reg}
  If as in Example~\ref{ex:int-reg} we represent the memory as a
  sequence, a memory filled with the value ``false'' would be
  $\mm_0=(0,0,0,\ldots)$. Assume  $c$ is the unary operation placing a fair
  coin at the corresponding address;  if $\mm_1$ is
  $\Mupdate(0,c,\mm_0)$,  we have $\mm_1=(\frac12,0,0,0,\ldots)$.  Then
  $\Mtest(0,\mm_1)$ is the distribution
  $\frac12({\rm false},(0,0,0,0,\ldots))+ 
  \frac12({\rm true},(1,0,0,0,\ldots))$.
\end{example}

\subsubsection{Quantum Registers}
\label{sec:quantum-memory}

A standard model for quantum computation is the QRAM model: quantum
data is stored in a memory seen as a list of (quantum)
registers, each one holding a qubit which can be acted upon. The model
supports  three main operations: creation
of a new register, measurement of a register, and application of
unitary gates on one or more registers, depending on the arity of the
gate under scrutiny. This model has been used extensively
in the context of quantum
lambda-calculi~\cite{SelingerValiron,DalLagoZorziLINEARITY,PaganiSV14}, with 
minor variations. The main choice to be made is whether
measurement is destructive (\ie, if one uses garbage collection) or
not (\ie, the register is not reclaimed).

\paragraph{A Canonical Presentation of Quantum Memory.}
To fix things, we shall concentrate on the 
presentation given in~\cite{PaganiSV14}. We briefly recall it. 
Given  $n$ qubits,  a memory is a
normalized vector in $(\mathbb{C}^2)^{\tensor n}$ (equivalent to a
ray). A linking function maps the position of each qubit in the list
to some pointer name. The creation of a new
qubit turns the memory $\phi\in(\mathbb{C}^2)^{\tensor n}$ into
$\phi\tensor\ket{0}\in(\mathbb{C}^2)^{\tensor (n+1)}$. The
measurement is destructive: if $\phi = \alpha_0{}q_0 + \alpha_1{}q_1$,
where each $q_b$ (with $b=0,1$) is normalized of the form
$\sum_i\phi_{b,i}\tensor\ket{b}\tensor\psi_{b,i}$, then measuring
$\phi$ returns $\sum_i\phi_{b,i}\tensor\psi_{b,i}$ with probability
$|\alpha_b|^2$. Finally, the application of a $k$-ary unitary
gate $U$ on $\phi\in(\mathbb{C}^2)^{\tensor n}$ simply applies the
 unitary matrix corresponding to $U$ on the vector $\phi$. The language comes
with a chosen set $\mathcal{U}$ of such gates.

\paragraph{Quantum Memory as a Nominal Set.}
The quantum memory  
  can be equivalently
presented using a memory structure: in the following we shall refer to
it as $\Qmem$.
The idea is  to  use nominal set to make precise the hand-waved
``pointer name'', and 
 formalize the idea of having a finite core of "in use"
qubits, together with an infinite pool of fresh qubits.
Let $\mathcal F_0$  be the set
of \mbox{(set-)maps} from $I$ (the infinite, countable set of
Sec.~\ref{sec:mem-struct}) to $\{0,1\}$ that have value $0$
everywhere except for a finite subset of $I$.
We have a {\em memory structure} as follows.

$I$ is the domain of the set-maps in
$\mathcal F_0$.
The {\em nominal set} $(\memories,\cdot)$ is defined with
$\memories = \mathcal H_0$, \ie\ the Hilbert space built from finite
(complex) linear combinations over $\mathcal F_0$, while the {\em group
  action} $(\cdot)$ corresponds to permutation of addresses:
$\sigma\cdot\mm$ is simply the function-composition of $\sigma$
with the elements of $\mathcal{F}_0$ in superposition in $\mm$.
The support of a particular
memory $\mm$ is finite: it consists of the set of addresses that are
not mapped to $0$ by some (set)-function in the superposition.
The {\em set of operations} $\mathcal L$ is the chosen set
$\mathcal{U}$ of unitary gates.
The {\em arity} is  the arity of the corresponding gate.

Finally, the {\em update} and {\em test} operations correspond
respectively to the application of an unitary gate, and to a
measurement followed by a (classical) boolean test on the result.  We
omit the formalization of these operations in the nominal set setting;
instead we show how this presentation in terms of nominal sets is
equivalent to the previous more canonical one. 

\paragraph{Equivalence of the Two Presentations.} Let $\mm\in\Qmem$.
We can always consider a finite subset of $I$, say
$I_0 = \{i_0\ldots i_n\}$ for some integer $n$ such that all other
addresses are fresh. As fresh values are $0$ in $\mm$, then $\mm$ is a
superposition of sequences that are equal to 0 on $I\setminus I_0$. Then  $\mm$ can
 be represented as ``$\phi\otimes\ket{000\ldots}$'' for some
(finite) vector $\phi$. We can omit the last $\ket{0000\ldots}$ and
only work with the vector $\phi$: we are back to the canonical
presentation of quantum memory. Update and test can then be defined on
the nominal set presentation through this equivalence.

\paragraph{Equations.}
Memory structures come with equations, which are indeed  
satisfied by quantum memories. Referring to
Sec.~\ref{sec:mem-struct} : (1) is simply renaming of qubits, (2)
is a property of applying a unitary, and (3) holds because of the equations corresponding
to the tensor of two unitaries or the tensor of a unitary and a
measurement (see Remark~\ref{rem:nonlocal}).

%

\begin{remark}\label{rem:modif} The quantum case makes clear why 
  $\memories$ appears in the codomain of test:  in general the measurement
  of a register collapses the global state of the memory (see Sec.~\ref{sec:quantum-computation}). The modified
  memory therefore has to be returned together with the result.
\end{remark}

\subsection{Overview of the Forthcoming Sections}
We use memory structures to encapsulate effects in
three different settings. In Sec.~\ref{sec:prog-net}, we enrich proof nets
with a memory, in Sec.~\ref{sec:msiam}, we enrich token machines with a
memory, while in Sec.~\ref{sec:pcfam}, we equip \PCFwo\ terms with a memory.
The construction is uniform for all the three systems, to which we refer as {\em operational systems}, as opposite to the\emph{ base rewrite systems} on top of which we build (see Sec.~\ref{overview}). 

\section{Program Nets and Their Dynamics}\label{sec:prog-net}
\newcommand{\Mlink}{{\rm link}}

In this section, we introduce \emph{program nets}. The base rewrite
system on which they are built is a variation\footnote{In this paper,
  reduction of the \bbox\ is not deterministic; there is
  otherwise no major difference with \cite{lics2015}.} of
\SMELLY\ nets, as introduced in \cite{lics2015}.  \SMELLY\ nets are
\MELL\ (Multiplicative Exponential Linear Logic) proof nets extended
with \emph{fixpoints} (Y-boxes) which model recursion, \emph{additive boxes} (\bbox es) which capture the if-then-else construct, and
a family of \emph{sync nodes}, introducing explicit synchronization
points in the net.

The \emph{novelty} of this section is the operational system which we
introduce in Sec.~\ref{sec:program-nets}, by means of our
parametric construction: given a memory structure and \SMELLY\ nets,
we define program nets and their reduction. We prove that program nets
are a PARS which satisfies the diamond property, and therefore
confluence and uniqueness of normal forms both hold. Program nets also
satisfy cut elimination, \ie\ deadlock-freeness of  nets
rewriting.

\subsection{Formulas}
The language of \emph{formulas} is the same as  for
\MELL. In this paper, we restrict our attention to the constant-only fragment, \ie:
$$
\formone \bnf \one\midd
\bot\midd \formone\otimes\formone\midd\formone\parr\formone\midd\,!\formone\midd\, ?\formone.
$$ 
The constants $1,\bot$ are the \emph{units}.
As usual, linear negation $(\cdot)\b$ is extended into an involution on all
formulas: $A\b\b\equiv A$, $1\b\equiv \bot $,
$(\formone\otimes\formtwo)\b\equiv \formone\b\parr \formtwo\b $,
$(!\formone)\b \equiv{} ?\formone\b$.  Linear implication is a defined
connective: $\formone\lin\formtwo\equiv\formone\b\parr\formtwo$. 
\emph{Positive formulas} $P$ and \emph{negative formulas}  $N$ are
 respectively defined as:
$\posfone \bnf \one \midd \posfone\otimes\posfone$,
and $\negfone \bnf \bot \midd \negfone\parr\negfone$.

\subsection{\SMELLY\ Nets}\label{sec:nets}
  A \SMELLY\ \emph{net}  is  a pre-net (\ie\ a well-typed graph) which fulfills a
 \emph{correctness criterion}.

\paragraph{Pre-Nets.} 
A \emph{pre-net}   is a labeled \emph{directed} graph $R$ built
over the alphabet of nodes represented in Fig.~\ref{probaNets}.

\noindent\textit{{Edges.}} Every edge in $R$ is labeled with a formula;
the label of an edge is called its \emph{type}.
We call those edges represented below (resp. above) a node symbol
\emph{conclusions} (resp. \emph{premises}) of the node.
We will often say that a node ``has a conclusion (premise) $\formone$''
as shortcut for ``has a conclusion (premise) of type $\formone$''.
When we need more precision, we explicitly distinguish an edge and its type
and we use variables such as $e,f$ for the edges.
Each edge is a conclusion of exactly one node and is a premise
of at most one node.  Edges which are not premises of any node are called the \emph{conclusions of the net}.
          
\noindent\textit{{Nodes.}}
 The sort of each node induces constraints on the number and the labels
of its premises and conclusions.
The constraints are graphically shown in Fig.~\ref{probaNets}.
A \emph{sync node} has
$n \in \mathbb{N}$ premises of types $\posfone_1, \posfone_2, \cdots, \posfone_n$ respectively
and $n$ conclusions of the same types $\posfone_1, \posfone_2, \cdots, \posfone_n$ as the premises,
where each $\posfone_i$ is a positive type.
A sync node with $n$ premises and conclusions is drawn as $n$ many
black squares connected by a line as in the figure. The total number
of $\one$'s in the $P_i$'s is called the {\em arity} of the sync node.

We call \emph{boxes} the nodes $\bot$, $! $, and $Y $. The leftmost
conclusion of a box is said to be \emph{principal}, while the other ones
are \emph{auxiliary}. The node $\bot$ has conclusion $\{\bot,\Gamma\}$
with $\Gamma \not=\emptyset$.  The exponential boxes $!$ and $Y$ have
conclusions $\{!A,?\Gamma\}$ ($\Gamma$ possibly empty).
To each $!$-box (resp. $Y$-box) is associated a content, \ie\ a pre-net of
conclusions $\{A,?\Gamma\}$ (resp. $\{A, ?A\b, ?\Gamma\}$).  To each {\bbox} are associated
\emph{a left and a right content}: each content is a pair
$(\botlk,S)$, where $\botlk$ is a new node that has no premise and one
conclusion $\bot$, and $S$ is a pre-net of conclusions $\Gamma$.   We represent
a box $b$ and its content(s) as in Fig.~\ref{probaBoxes}. The
nodes and edges in the content are said to be \emph{inside} $b$. As is
standard, we often call a crossing of the box border a \emph{door},
which we treat as a node. We then speak of premises and
conclusion of the principal (resp. auxiliary) door.  \condinc{}{ Observe that in
the case of \bbox, the principal door has a left and a right premise.}

\noindent\textit{Depth.} A node occurs \emph{at depth 0} or \emph{on the surface} in the
 pre-net $R$ if it is a node of $R$. It occurs \emph{at depth $n+1$} in $R$ if it occurs at depth $n$ in a pre-net associated to a box of $R$.  
\begin{figure}
\centering
  \centering
  {\includegraphics[width=\columnwidth]{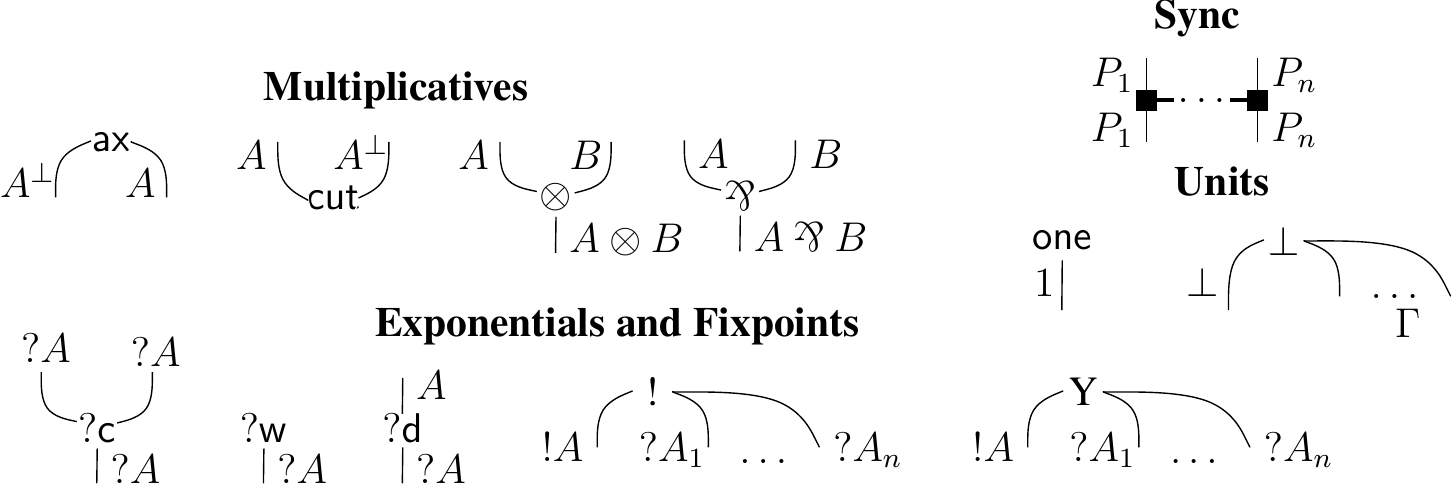}}
  \caption{Nodes.}
  \label{probaNets}
\end{figure}  

\begin{figure}
{\includegraphics[width=\columnwidth]{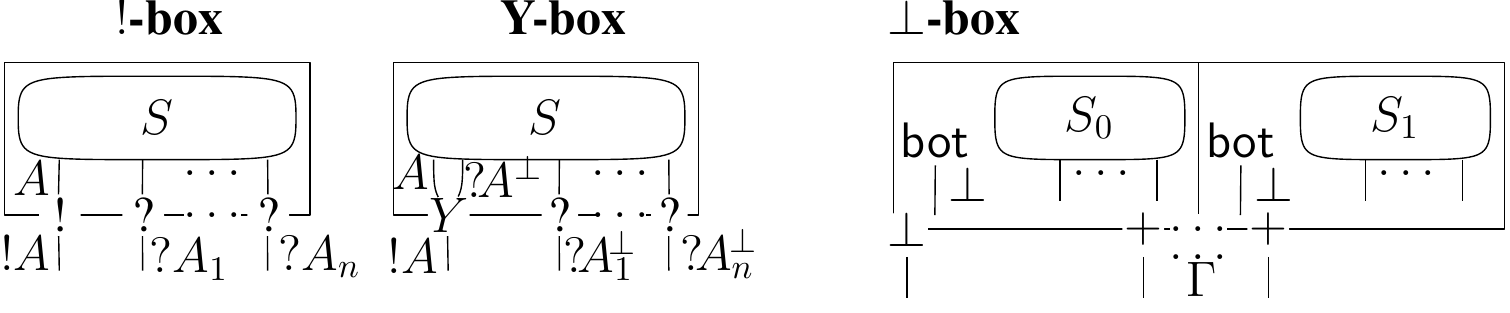}}
    \caption{Boxes.}
    \label{probaBoxes}
\end{figure}

\begin{figure}
  \centering
  {\includegraphics[width=\columnwidth]{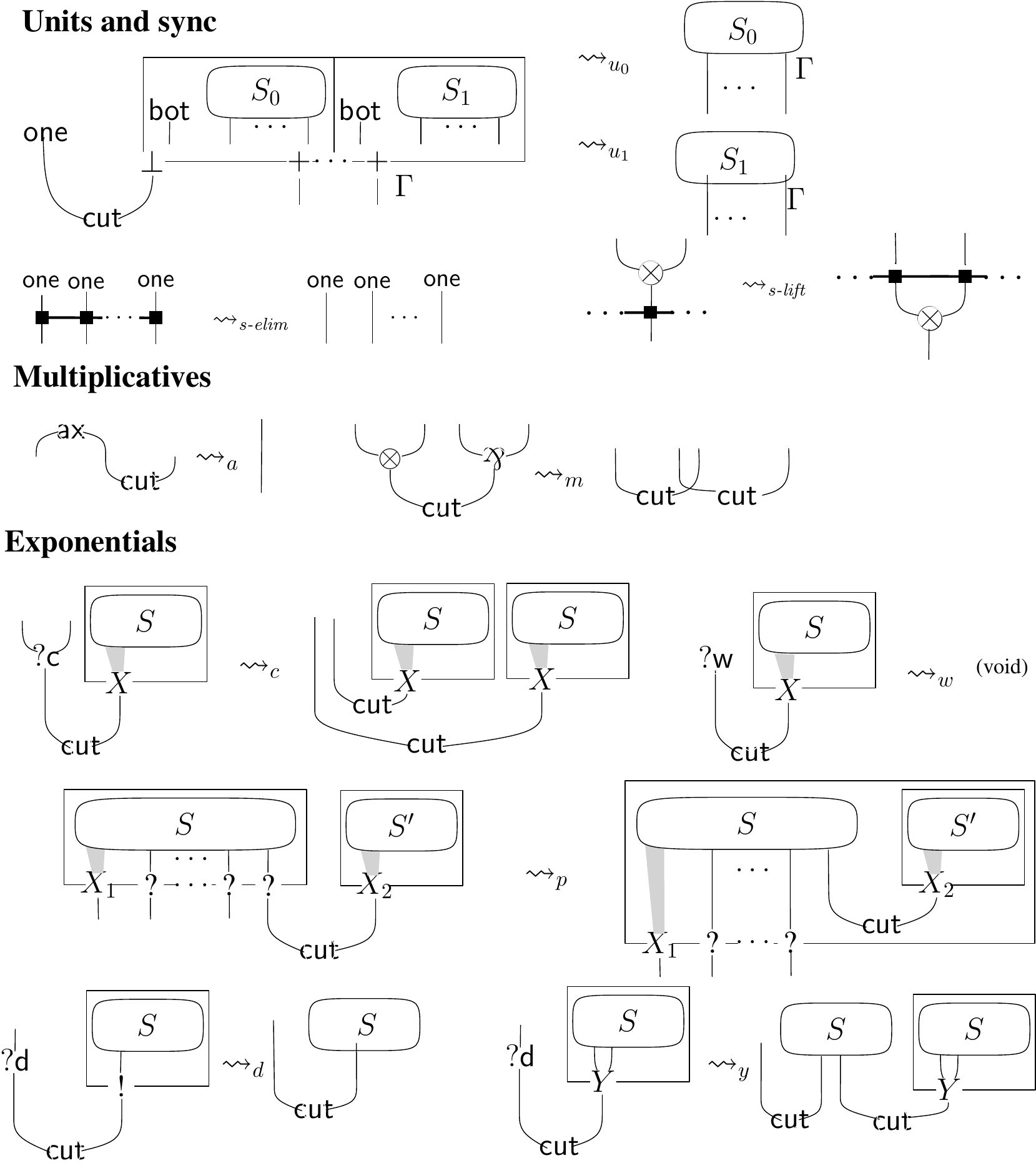}}
     \caption{Nets Rewriting Rules.}
     \label{probaRed}
\end{figure}
\paragraph{Nets.}
A \emph{net} is given by a pre-net $R$ which satisfies the
correctness criterion of~\cite{lics2015}, together with a {\em total} map
$\mapsyncname{R}:\syncnode{R}\to \mathcal L$,
where $\mathcal L$ is a finite set of \emph{names}
and $\syncnode{R}$ is the set of sync nodes appearing in $R$ (including those inside boxes);
 the map $\mapsyncname{R}$ is simply naming the sync nodes.
From now on, we   write  $R$ for the triple $(R, \mathcal{L},\mapsyncname{R})$.
     
\condinc{}{Correctness is defined by means of switching paths.
A \emph{switching path} in the pre-net $R$
is an undirected path on the graph $R$ (i.e.\ $R$ is regarded as an undirected graph) which uses at most one of the two \emph{premisses} for each $\parr$ and $?c$ node,
and at most one of conclusions for each sync node.
A pre-net   is \emph{correct} if 
  none of its switching paths is cyclic, and
  the content of each of its box is itself correct. }

\paragraph{Reduction Rules.}
Fig.~\ref{probaRed} describes the rewriting  rules on nets.
Note that the  redex in the top row   has two possible reduction rules, $u_0$ and $u_1$. Note also the $y$ reduction, which captures the recursive behavior of the  $Y$-box as a fixpoint (we illustrate this in the example below.)
   The metavariables  $X,X_1,X_2$ of Fig.~\ref{probaRed} range over $\{!,Y\}$
  and are used to uniformly specify reduction rules involving
  exponential boxes (\ie, $X$'s can be either $!$ or $Y$).
  The reduction of net has two constraints: (1.)
  \emph{surface reduction},
  \ie\ a reduction step applies only when the 
  redex is at depth 0,
  and  (2.) \emph{exponential steps are closed}, \ie\ they only take place
  when the $!A$ premise of the cut is the principal conclusion of a  box with no auxiliary conclusion. We come back on the former in Sec.~\ref{why surface}.
  
\condinc{}{ As expected, the net reduction preserves correctness.}
  

\paragraph{Example.}
The ``skeleton'' of the program in  Example~\ref{ex:recursion} could  be encoded\condinc{}{\footnote{%
Precisely speaking, the nets shown in Fig.~\ref{fig:recursion-net-1}
are those obtained by the translation given in
Fig.~\ref{fig:PCFtoNET1} and~\ref{fig:PCFtoNET2} (in the Appendix),
with a bit of simplification for clarity of discussion and due to lack
of space.%
  }} as in the LHS of Fig.~\ref{fig:recursion-net-1}.  The recursive
  function {\tt f} is represented with a Y-box, and has type
  $!(1\multimap1) = {!}(\bot\parr1)$. The test is encoded with a
  \bbox: in one case we forget the function {\tt f} by using
  $\mathsf{?w}$ and simply return a $\onelk$ node, and in the other
  case we apply a Hadamard gate, which  is represented with a (unary)
  sync node. To the ``{\tt in}''
  part of the let-rec corresponds the dereliction node $\mathsf{?d}$,
  triggering reduction. With the rules presented in the previous
  paragraph, the net rewrites according to
  Fig.~\ref{fig:recursion-net-1}, where the Y-box has been unwound
  once. From there, we could reduce the $\onelk$ node with the sync
  node $\Had$, but of course doing so we would not handle the quantum memory.
  In order to associate to reductions an action on the memory, we need a bit more, namely
  the notion of a program net which is introduced in
  Sec.~\ref{sec:program-nets}.

\begin{figure}
  \centering
  \includegraphics[page=2,width=\columnwidth]{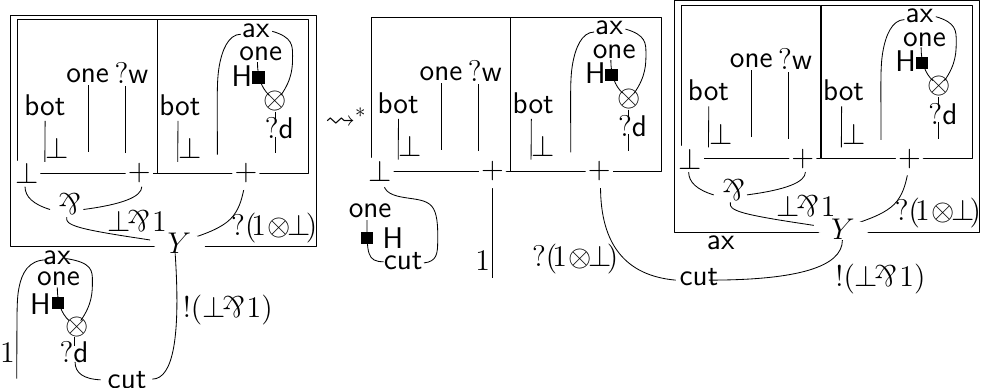}
  \caption{Encoding of Example~\ref{ex:recursion} and Net Rewrite.}
  \label{fig:recursion-net-1}
\end{figure}

\subsection{Program Nets }\label{sec:program-nets}

Let $\inputs{R}$ be the set of all occurrences of $\one$'s which are
conclusions of $\onelk$ nodes \emph{at the surface}, and of all the
occurrences of $\bot$'s which appear in the conclusions of $R$.

\begin{definition}[Program Nets]\label{def:prognets}
  Given a memory structure
  $\memories = (\memories,I,\mathcal L)$,
  a {\em raw program net} on $\memories$
  is a tuple $(R,\ind{R},\mm)$ such that
  \begin{varitemize}
  \item
    $R$ is a \SMELLY\ net (with $\mapsyncname{R}:\syncnode{R}\to
    \mathcal L$),
  \item
    $\ind{R}:\inputs{R}\partto I$ is an injective \emph{partial} map
    that is however total on the occurrences of $\bot$,
  \item
    $\mm\in \memories$.
  \end{varitemize}
   We require that the arity of each sync node $s$ matches the arity
   of $\mapsyncname{R}(s)$.  Please observe that in the second item in
   Definition~\ref{def:prognets}, the occurrences of $\one's$
   belonging to $\inputs{R}$ are not necessarily in the domain of
   $\ind{R}$; if they are, we say that the corresponding $\onelk$ node
   is \emph{active}.  \emph{Program nets} are the equivalence class of
   raw program nets over permutation of the indexes.  Formally, let
   $\sigma(R, \ind{R}, \mm) = (R, \sigma\cdot\ind{R}, \sigma \cdot
   \mm)$, for $\sigma\in\Finbij(I)$.  The equivalence class $\R =
      [(R,\ind{R}, \mm)]$ is $\{ \sigma(R, \ind{R}, \mm) ~|~
      \sigma\in\Finbij(I) \} $.
 %
   We use the symbol $\sim$ for the equivalence relation on raw
   program nets.  $\mathcal{N}$ indicates the set of program nets.
\end{definition}

\paragraph{Reduction Rules.}
We define a relation ${\rednet} \subseteq \prognetsset \times
\dists\prognetsset$, making program nets into a PARS.  We first define
the relation $\rednet $ over raw program nets.  Fig.~\ref{fig:unitRed}
summarizes the reductions in a graphical way; the function $\ind{R}$
is represented by the dotted lines.
\begin{figure}
  \centering
    \includegraphics[width=\columnwidth]{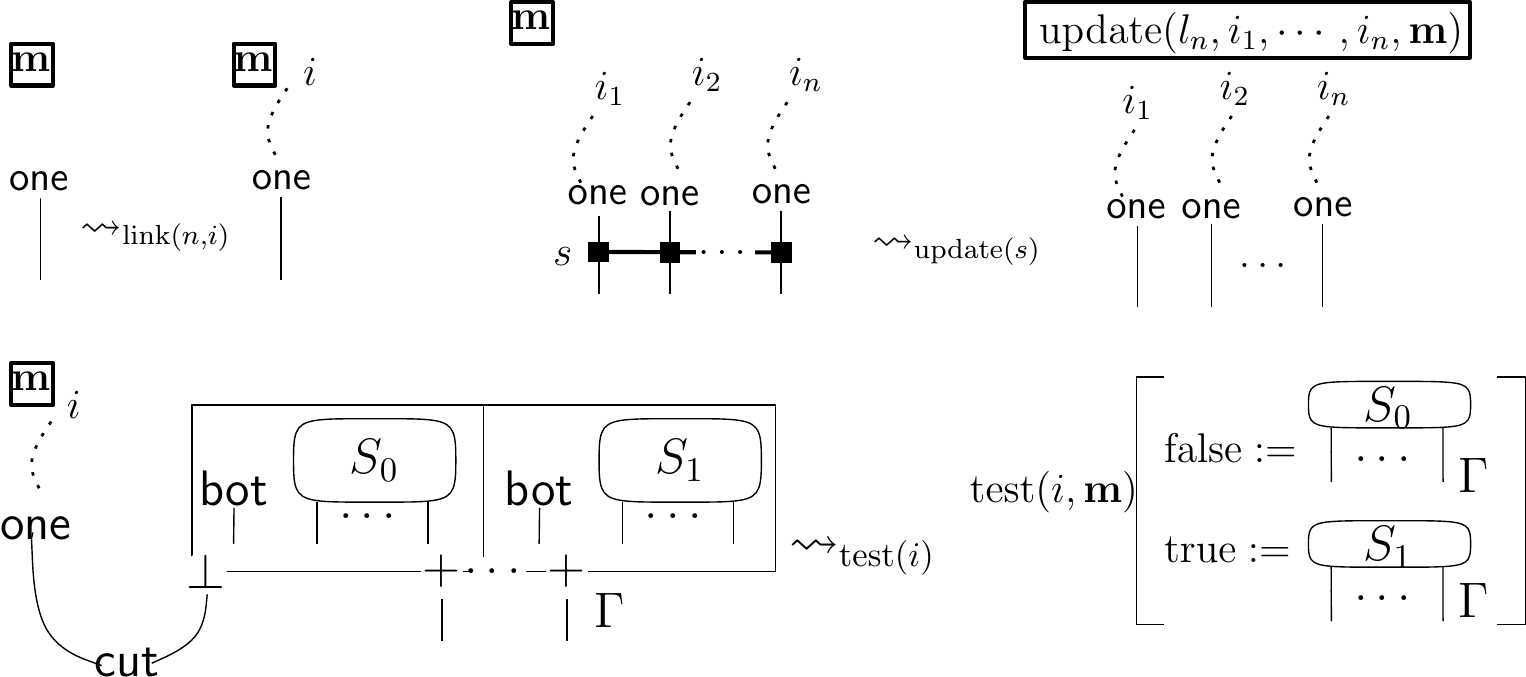}
  \caption{Program Net Reduction Involving Memories.}\label{fig:unitRed}
\end{figure}
\begin{varenumerate}
\item\label{item:linkRed}
  \textit{Link}. If $n$ is a $\onelk$ node of conclusion $x$, with
  $\ind{R}(x)$ undefined, then
    $(R,\ind{R},\mm)
    \rednet_{\mathrm{link(n,i)}} \{(R, \ind{R}\cup\{x\mapsto i\},\mm)^1\} $ where    $i\in I$ is 
        fresh both in $\ind{R}$ and $\mm$.
\item\label{item:syncRed}
   \textit{Update}.  If $R\rednet_s R'$, and $s$ is the sync node in
   the redex, then $$(R, \ind{R}, \mm) \rednet_{\mathrm{update(s)}}
   \{(R', \ind{R}, \upd(l,\vec i,\mm)^1\}$$ where $l$ is the label of
   $s$, and $\vec i$ are the addresses of its premises.
 \item\label{item:botRed}
   \textit{Test}.  If $R\rednet_{u_0} R_0$ and $R\rednet_{u_1}R_1 $,
   and $i$ is the address of the premise $\one$ of the cut, then $(R,
   \ind{R}, \mm) \rednet_{\mathrm{test(i)}}$
   $\test(i,\mm)[\Mfalse{:}{=}(R_0, \ind{R_0}, \mm),$
     $\Mtrue{:}{=}(R_1,$ $ \ind{R_1}, \mm)]$, where $\ind{R_0}$
   (resp.\ $\ind{R_1}$) is the restriction of $\ind{R}$ to
   $\inputs{R_0}$ (resp.\ $\inputs{R_1}$).
 \item\label{item:MELLRed}
   Otherwise, if $R\rednet_x R' $ with $x\not\in \{ s, u_0, u_1\}$,
   then we have $$(R, \ind{R}, \mm)\rednet_x \{(R', \ind{R},
   \mm)^1\}$$ (observe that none of these rules modify the domain of
   $\ind{R}$).
\end{varenumerate}
  The relation $\rednet$ extends immediately to program nets (by slight
abuse of notation we use the same symbol); Lem.~\ref{lemma: equiv}
guarantees that the relation is well defined.  We write $(R, \ind{R},
\mm) \stackrel{r}{\rednet} \mu$ for the reduction of the redex $r$ in
the raw program net $(R, \ind{R}, \mm)$.

\begin{lemma}[Reduction Preserves Equivalence]\label{lemma: equiv}
  Suppose that $(R, \ind{R}, \mm) \stackrel{r}{\rednet} \mu$ and
  $(R, \sigma\cdot\ind{R}, \sigma \cdot \mm) \stackrel{r}{\rednet}
  \nu$, then $\mu \sim \nu$.
\end{lemma}

\begin{proof}
  Let us check the rule $\rednet_{\Mtest(i)}$. 
  Suppose $(R',\ind{R'},\mm') = \sigma(R,\ind{R},\mm)$,
  $(R, \ind{R}, \mm) \stackrel{r}{\rednet}_{\Mtest\!(\!i\!)}
  \mu\,{=}\, \{  (R_0, \ind{R_0}, \mm_0)^{p_0}\!,$
  $(R_1, \ind{R_1}, \mm_1)^{p_1}  \}$,
  and $(R', \ind{R}', \mm') \stackrel{r}{\rednet}_{\Mtest(i)}$

\noindent
  $\nu = \{  (R_0', \ind{R_0}', \mm_0')^{p_0'},
  (R_1', \ind{R_1'}, \mm_1')^{p_1'}  \}$ by reducing the same redex $r$.
  It suffices to show that $\nu = \sigma \Mswap \mu$.
  Element-wise, we have to check
  $R_i' = R_i$, $\ind{R_i}' = \sigma \circ \ind{R_i}'$, $\mm_i' = \sigma \Mswap \mm_i$, and $p_i' = p_i$ for $i \in \{0,1\}$.
  The first two follow by definition of $\rednet_{\Mtest(i)}$ and
  the last two follow from the equation $\sigma\Mswap(\Mtest(i,m)) = \Mtest(\sigma(i), \sigma\Mswap m)$.
  The other rules can be similarly checked.
\end{proof}

\begin{remark}
  In  the definition of the  reduction rules: 
  \begin{varitemize}
  \item \textit{Link}
    is independent from  the choice of $i$.
    If we chose another address $j$ with the same conditions,
    then we would have gone to
    $(R, \ind{R}\cup\{x\mapsto j\},
     \mm)$.
    However this does not cause a problem:
    by using a permutation $\sigma = (i~ j)$,
    since $\sigma\Mswap\mm = \mm$ we
    have
    $
    \sigma(R, \ind{R}\cup\{x\mapsto i\},\mm)
    =
    (R, \ind{R}\cup\{x\mapsto j\},\mm)
    $
    and therefore $(R, \ind{R}\cup\{x\mapsto i\}, \mm)$ and
    $(R, \ind{R}\cup\{x\mapsto j\},  \mm)$
    are as expected the exact same program net.
  \item  In the rules \textit{Update} and \textit{Test}, 
    the involved  $\onelk$ nodes  are required to be active.
  \end{varitemize}
\end{remark}

The pair $(\mathcal{N},\rednet)$ forms a PARS.
Reduction can happen at different places in a net, however the diamond property allows us to deal with this seamlessly.

\begin{figure}
  \centering
  \includegraphics[width=\columnwidth,page=4]{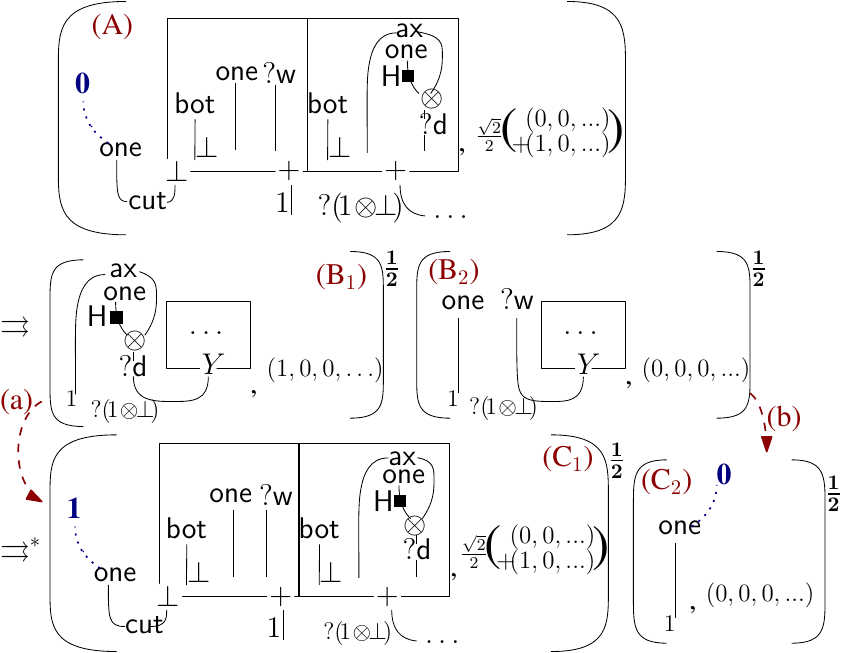}
  \caption{Example of Program Net Rewriting.}
  \label{fig:choice-program-net}
\end{figure}

\begin{proposition}[Program Nets are Diamond]\label{prop:netDiamond}
The PARS $(\mathcal{N}, \rednet)$ satisfies the diamond property.\qed
\end{proposition}
The proof \condinc{}{(see Appendix \ref{app:nets}) }relies on
commutativity of the memory operations.
Due to Th. \ref{th:proba_term}, program net reduction enjoys all
the good properties we have studied in Sec. \ref{sect:confluence}:
\begin{corollary}
  The relation $\rednet$ satisfies Confluence, Uniformity and
  Uniqueness of Normal Forms (see Th. \ref{th:proba_term}).
  \qed
\end{corollary}
The following two results can be obtained as adaptations of
similar ones in  \cite{lics2015}.
\begin{theorem}[Deadlock-Freeness of Net Reduction]\label{main_lem} 
  Let $\R = [(R,\ind{R},$ $  \mm)]$ be a program net such that no $\bot$, $?$ or $!$
  appears in the conclusions of the net $R$.  If $R$ contains cuts, a
  reduction step is always possible.\qed
\end{theorem}
\begin{corollary}[Cut Elimination] \label{cutel} With the same
  hypothesis as above, if $\R\not\rednet$, (\ie\ no further reduction
  is possible) then $\R$ is cut free.\qed
\end{corollary}

\paragraph{Example.}
The net in LHS of Fig.~\ref{fig:recursion-net-1} can be embedded into
a program net with a quantum memory of empty support:
$(0,0,0,\ldots) \equiv{}$ ``$\ket{000\ldots}$''. It reduces according
to Fig.~\ref{fig:recursion-net-1}, with the same memory. The next
step requires a $\rednet_{\mathrm{link}}$-rewrite step to attach a fresh
address---say, $0$---to the $\onelk$ node at surface. The
$\mathsf{H}$-sync node then rewrites with a $\rednet_{\mathrm{update}}$-step, and we
get the program net (A) in Fig.~\ref{fig:choice-program-net} with
the ``update'' action applied to the memory: the memory corresponds to
$\frac{\sqrt2}2(\ket0+\ket1)\otimes\ket{00\ldots}$. From there, a
choice reduction is in order: it uses the ``test'' action of the
memory structure, which, according to Sec.~\ref{sec:quantum-memory}
corresponds to the measurement of the qubit at address $0$. This
yields the probabilistic superposition of the program nets (B${}_1$)
and (B${}_2$). As the net in (B${}_1$) is the LHS of
Fig.~\ref{fig:recursion-net-1}, it reduces to (C${}_1$) (dashed arrow
(a)), similar to (A) modulo the fact that the address $0$ was not
fresh: the $\rednet_{\mathrm{link}}$-rewrite step cannot yield $0$: here we
choose $1$. Note that we could have chosen any other non-zero number as the address.
The program net (B${}_2$) rewrites to (C${}_2$) (dashed arrow (b)):
the weakening node erases the Y-box, and a fresh variable is
allocated. In this case, the address 0 is indeed fresh and can be
picked.

\section{A Memory-Based  Abstract Machine}\label{PSIAM}\label{sec:msiam}
In this section we introduce a class of memory-based token machines,
called the \PSIAM\ (\emph{Memory-based Synchronous Interaction Abstract Machine}).
The base rewrite system on which the \PSIAM\ is built, is a
variation\footnote{In this paper we make this transition
  non-deterministic; otherwise there is no major difference.} of the
\SIAM\ multi-token machine from \cite{lics2015}, which we recall in
Sec. \ref{sec:siam}.  The specificity of the \SIAM\ is to allow not
only parallel threads, but also \emph{interaction} among them,
\ie\ \emph{synchronization}. Synchronization happens in particular at
the sync nodes (unsurprisingly, as these are  nodes
introduced with this purpose), but also on the additive boxes (the
\bbox). The transitions at the \bbox\ model \emph{choice}: as we see
below, when the flow of computation reaches the \bbox\ (\ie\ the
tokens reach the auxiliary doors), it continues on one of the two
sub-components, depending on the tokens which are positioned at the
principal door.

The original contribution of this section is contained  in Sections
~\ref{sec:psiamDef} through ~\ref{semantics}, where we use our
parametric construction to define the \PSIAM\ $\machine{\R}$ for $\R$
as a PARS consisting of 
	a set of \emph{states} $\sts$, and
	a \emph{transition relation} $\mathord{\redsiam} \subset \sts \times \dists{\sts}$,
and establish its main properties, in particular Deadlock-Freeness
(Th.~\ref{deadlock-free}), Invariance (Th.~\ref{invariance})
and Adequacy (Th.~\ref{adequacy psiam}).

\subsection{\SIAM}\label{sec:siam}
Let $R$ be a net. The \SIAM{} for $R$ is given by a set of states and a transition relation  on states.
Most of the definitions are standard.

\emph{Exponential signatures} $\sigma$  and \emph{stacks} $s$ are defined by\\ 
$
\begin{array}{l@{~}l@{~}l}
\sigma \bnf {*} \midd l(\sigma) \midd r(\sigma)
\midd \encode{\sigma}{\sigma} \midd y(\sigma,\sigma)\\
\stk \bnf \emp \midd l.\stk \midd r.\stk \midd \sigma.\stk \midd \delta
\end{array}
$\\
where $\epsilon$ is the empty stack and $.$ denotes concatenation.
Two kinds of stacks are defined: (1.) the \emph{formula stack} and
(2.) the \emph{box} stack. The latter is the standard GoI way
to keep track of the different copies of a box.  The former describes
the formula path of either an occurrence $\atomone$ of a unit, or an
occurrence $ \Diamond$ of a modality, in a formula $A$.
Formally, $\stk$ is a \emph{formula stack on $A$} if either
$\stk=\delta$ or $\stk[A]= \atomone$ (resp. $\stk[A]= \Diamond$), with
$\stk[A]$ defined as follows: $\emp [\atomone] = \atomone$,
$\sigma.\delta[\Diamond B] = {\Diamond}$, $\sigma.t[\Diamond \typetwo]
= t[\typetwo] $ whenever $t\neq\delta$, $l.t[\typetwo \Box
  \typethree]= t[\typetwo] $ and $r.t[\typetwo \Box
  \typethree]=t[\typethree]$ (where $\Box$ is either $\otimes$ or
$\parr$).  We say that $\stk$ \emph{indicates} the occurrence
$\atomone$ (resp. $ \Diamond$). 
\begin{example}\label{ex:indic}
  Given the formula $A =\bang(\bot\otimes{!\one})$, the stack $*.\delta$
  indicates the \emph{leftmost} occurrence of $\bang$,
  the stack $ {*}.r.{*}.\delta$ indicates the
  \emph{rightmost} occurrence of $!$, and $ *.l[A]=\bot$.
\end{example}
\paragraph{Positions.}
Given a net $R$, the set of its \emph{positions} $\POSALL_R$ contains all
the triples $(\edg, \stk, \bstk)$, where $\edg$ is an edge of
$\netone$, $\stk$ is a formula stack on the type $A$ of $\edg$, and
$\bstk$ (the \emph{box stack}) is a stack of $n$ exponential
signatures, where $n$ is the depth of $\edg$ in $R$. We use the
metavariables $\ss$ and $\pp$ to indicate positions. For each position
$\pp=(\edg,\stk,\bstk)$, we define its \emph{direction} $\dr (\pp)$ to
be \emph{upwards} ($\up$) if $\stk$ indicates an occurrence of $!$ or
$\bot$, to be \emph{downwards} ($\down$) if $\stk$ indicates an
occurrence of $?$ or $\one$, to be \emph{stable} ($\stable$) if $\stk=
\delta$ or if the edge $\edg$ is the conclusion of a $\botlk$ node.
The following subsets of $\POSALL_R$ play a  role in the
definition of the machine:
\begin{varitemize}
\item 
  the set $\POSI_R$ of \emph{initial positions}  $\pp=(\edg, \stk, \emp)$,
      with $e$  conclusion of
  $\netone$, and $\dr(\pp)$ is $ \up$;
\item 
  the set $\POSF_R$ of  \emph{final positions} $\pp=(\edg, \stk, \emp)$, with $e$  conclusion of
    $\netone$, and $\dr(\pp)$ is $ \down$;
\item 
  the set $\ONES_{\netone}$ of positions $(\edg, \emp, \bstk)$,
  $\edg$ conclusion of a $\onelk$ node;
\item 
  the set $\DEREL_{\netone}$ of positions $(\edg, *.\delta, \bstk)$,  $\edg$  conclusion of a $\dlk$ node;
\item   
  the set $\PDOORS_R$ of the positions $\pp$ for which
  $\dr(\pp) = {\stable}$;
  \item the set of starting positions 
    $\START_{\netone}= \POSI_R \cup
    \ONES_{\netone}\cup \DEREL_{\netone} $. 
\end{varitemize}

\paragraph{\SIAM\ States.}
A \emph{state} $(T, \orig)$ of $\M_R$ is a set of positions $T\subseteq
\POSALL_R$
equipped with an injective map $\orig: T \to
\START_R$.  Intuitively, $T $ describes the current positions of
the tokens, and $\orig$ keeps track of where each such token started its path.

 A state  is \emph{initial} if $ T\subseteq \POSI_R$ and $\orig$ is
 the identity. We indicate the (unique) initial state of $\M_R$ by
   $I_R$. A state $T$ is \emph{final} if all positions in $T$
   belong to either $\POSF_R$ or $\PDOORS_R$.
   
   With a slight abuse of notation, we will denote the state $(T,\orig)$ also by $T$.
   Given a state $T$ of $\M_R$, we say
   that \emph{there is a token in $\pp$} if $\pp\in T$.  We use
   expressions such as ``a token moves'', ``crosses a node'', in the
   intuitive way.

\paragraph{\SIAM\ Transitions.}
The  transition rules of the \SIAM\ are described in Fig.~\ref{fig:trRules bot} and \ref{fig:trRules}. 
Rules (i)-(iv) require synchronization among different tokens; this is expressed by specific multi-token conditions
which we discuss in the next paragraph.   First, we  explain the graphical conventions and give an overview of the rules.
\begin{figure}
  \centering
  \includegraphics[width=.85\columnwidth]{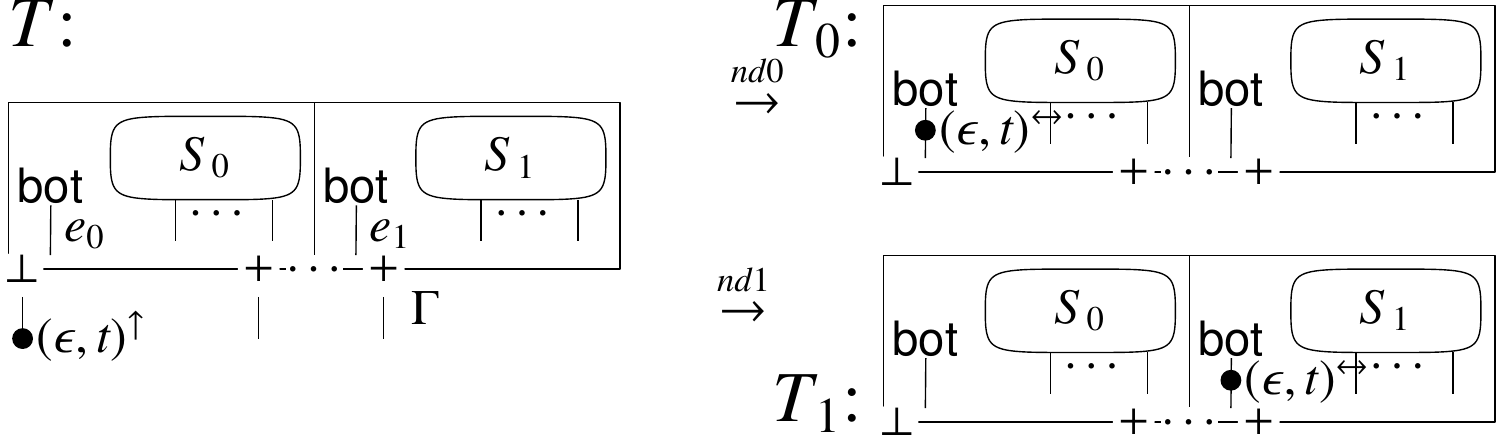}
  \caption{\SIAM\: Non-Deterministic Transition Rules.} 
  \label{fig:trRules bot}
\end{figure}
\begin{figure}
  \centering
  \includegraphics[width=\columnwidth]{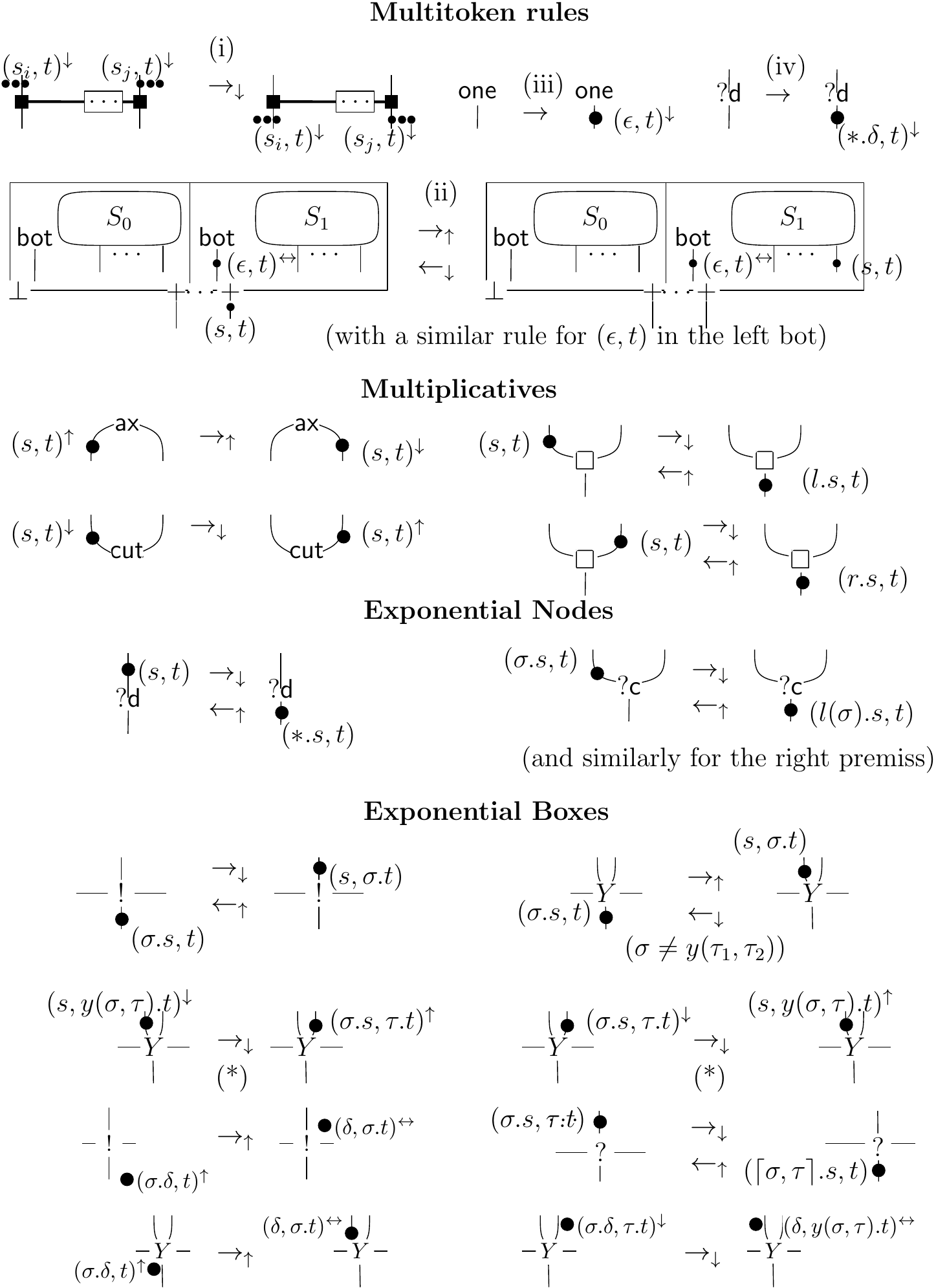}
  \caption{\SIAM: Deterministic Transition Rules.}
  \label{fig:trRules}
\end{figure}

The position $\pp = (\edg, \stk, \bstk)$ is represented graphically by
marking the edge $e$ with a bullet $\bullet$, and writing the stacks
$(\stk, \bstk)$.  A transition $T \red U$ is
given by depicting only the positions in which $T$ and $U$
differ. 
 It is  intended that all positions of $T$ which do
not explicitly appear in the picture also belong to $U$.  To
save space, in Fig.~\ref{fig:trRules}  the transition
arrows are annotated with a \emph{direction}; this means that the rule applies (only)
to positions which have that direction.  When useful, the direction of a position is directly annotated  with
${}^{\down},{}^{\up}$ or ${}^{\stable}$. Note that no transition is
defined for stable positions.  For boxes,  whenever a token is on a conclusion of
a box, it can move into that box (graphically, the token ``crosses''
the border of the box) and it is modified as if it were crossing a
node. For exponential boxes, Fig.~\ref{fig:trRules} depicts only
the border of the box.
We do not explicitly give the function $\orig$, which is immediate to reconstruct when keeping in mind that it is a pointer to the origin of the token.

We briefly discuss  the most interesting transitions (we refer to
\cite{lics2015} for a broader discussion).  \emph{Fixpoints:} the
recursive behavior of Y-boxes is captured by the exponential signature
in the form $y(\cdot,\cdot)$, and the associated
transitions. \emph{Duplication:} the key is rule (iv), which generates a
 token on the conclusion of a $\dlk$ node; this token
will then travel the net, possibly crossing a number of contractions,
until it finds its exponential box; intuitively, each such 
token corresponds to \emph{a copy of a box}.  \emph{One:} the behavior
of the token generated by rule (iii) on the conclusion of a $\onelk$
node is similar to that of a dereliction token; the $\onelk$ token
searches for its \bbox. \emph{Stable tokens:} when the token from {an
  instance} of a $\dlk$ or of a $\onelk$ node ``has found its box'',
\ie\ it reaches the principal door of a box, the token become
\emph{stable} ($\stable$). A stable token is akin to a
marker sitting on the principal door of its box, 
keeping track of the box copies and of the choice made in each specific
copy.

\paragraph{Multi-token Conditions.} 
The rules marked by (i), (ii), (iii), and (iv) in
Fig.~\ref{fig:trRules} require the tokens to interact,  which is
formalized by  multi-token conditions. Such conditions allow, in particular, to capture choice
and synchronization.  
Below we give
 an intuitive presentation; we refer to \cite{lics2015,LVarxiv}
for the formal details. \condinc{}{For convenience we also recall
  them in Appendix \ref{sec:multi-token}.}
 
\emph{Synchronization, rule (i)}.  To cross a sync node $l$, all the positions on the premises of $l$ (for the  same box stack $\bstk$) must be  filled; intuitively, having the same $\bstk$, means that
the positions all belong to the same  copy of $l$. Only when all the tokens have reached $l$,  they can cross it; they do so   simultaneously.
%
 
\emph{Choice, rule (ii)}. Any token arriving at a \bbox\ on an
auxiliary door must wait for \emph{a token on the principal door to
  have made a choice} for either of the two contents, $S_0$ or
$S_1$:  a token $(e,s,t)$ on the conclusions $\Gamma$
of the \bbox\ will  move to $S_0$ (resp. $S_1$)  \emph{only if}   the principal door of $S_0$ (resp. $S_1$) has a token with the same 
$\bstk$. 
    
The rules marked by \emph{(iii) and (iv)} also carry a multi-token
condition, but in a more subtle way: a token is enabled to start its
journey on a $\onelk$ or $\dlk$ node only when its box has been
opened; this  reflects in the \SIAM\  
the constraint of   \emph{surface reduction}
of nets.

\subsection{\PSIAM}\label{sec:psiamDef}
Similarly to what we have done for nets, we enrich the machine with a memory, and use the \SIAM\ and the operations on the memory to define a PARS.

\paragraph{\PSIAM\ States.}
Given a memory structure $\memories =(\memories, I, \mathcal L)$ and 
  a raw  {program net} $( R,\ind{R} ,\mm_R)$ on $\memories$, 
 a \emph{ raw state} of the \PSIAM\ $\machine{\R}$ is 
  a tuple $(T, \ind{\st}, \mm_T)$ where 
  \begin{varitemize}
  \item $\stT$ is a state of $\M_R$,
  \item  $\ind{\stT}\colon \START_R \to I$ is a partial  injective map,
  \item $\mm_T \in \memories$.
  \end{varitemize}
 States    are defined as the equivalence class  $\stTT =[(\stT,\ind{\stT}, \mm_T)]$ of row states \emph{over permutations}, with 
    the  action  of   $\perm(I)$ on  tuples  being the natural one. 
 
    \paragraph{\PSIAM\ Transitions.} Let $\R$ be a program net, and
    $\stTT$ be a state $[(\stT,\ind{\stT}, \mm_T)]$ of $\M_{\R}$. We
    define the transition
    $\st \red \mu \in \sts \times \dists{\sts} $.  As we did for
    program nets, we first give the definition on raw states.  The
    definition depends on the \SIAM\ transitions for $T$.  Let us
    consider the possible cases.

 \begin{varenumerate}
 \item
   \textit{Link}. Assume $\stT \stackrel{(iii)}\red \stU$
   (Fig.~\ref{fig:trRules}), and let $n$ be the $\onelk$ node, $x$ its
   conclusion, and $\pp$ the new token in $\stU$.  We set
    $$(\stT,\ind{\stT},\mm)    \red_{\mathrm{link(n,i)}} 
    (\stU, \ind{\stT}\cup\{ \orig(\pp)\mapsto i\},\mm)$$
    where we choose  $i=\ind{R}(x)$ if the  $\onelk$ node is active,  and otherwise an address   $i$  which is fresh for both  $\ind{\stT}$ and $(\mm)$.
  \item
    \textit{Update}.  Assume $\stT \stackrel{(i)}\red \stU$
    (Fig.~\ref{fig:trRules}), $l$ is the name associated to the sync
    node, and $\vec i$ are the addresses which are associated to its
    premises (by composing $ \orig{}$ and $\ind{}$), then $$(\stT,
    \ind{\stT}, \mm) \red_{\mathrm{update(s)}} \{(\stU, \ind{\stT},
    \upd(l,\vec i,\mm)^1\}.$$
  \item
    \textit{Test}.  Assume $\stT \stackrel{nd 0}\red \stT_0$ and $\stT
    \stackrel{nd 1}\red \stT_1 $ (Fig.\ \ref{fig:trRules bot},
    non-deterministic transition).  If $\pp\in \stT$ is the token
    appearing in the redex (Fig.~\ref{fig:trRules bot}), and $i$ the
    addresses that $\ind{\stT}$ associates to $\orig(\pp)$,
    then
    \begin{center} $(\stT, \ind{\st}, \mm)
      \red_{\mathrm{test(i)}}$\\ $\test(i,\mm)[\Mfalse:=(\stT_0,
        \ind{\stT}), \Mtrue:=(\stT_1,$ $ \ind{\stT})]$.
    \end{center}
  \item
    In all the other cases: if $\stT\red \stU$ then $(\stT,
    \ind{\stT}, \mm) \red \{(\stU, \ind{\stT}, \mm)^1\}.$
\end{varenumerate}

Let $\R=[(R,\ind{R},\mem{R})]$. The \emph{initial state} of 
 $\machine{\R}$ is  $\stII_{\R}=[(I_R, \ind{I_R},\mem{R})]$, where  $ \ind{I_R}$ is only defined on the initial positions:
 if $\pp\in \POSI_R $, and $x$ is the occurrence of $\bot$ corresponding to $\pp$, then 
 $\ind{I_R}(\pp) =\ind{R}(x)$. 
A state $[$$(T, \ind{\stT},\mem{\stT})$$]$ is \emph{final} if $T$ is final.

\begin{figure}
  \centering
  \includegraphics[width=\columnwidth,page=5]{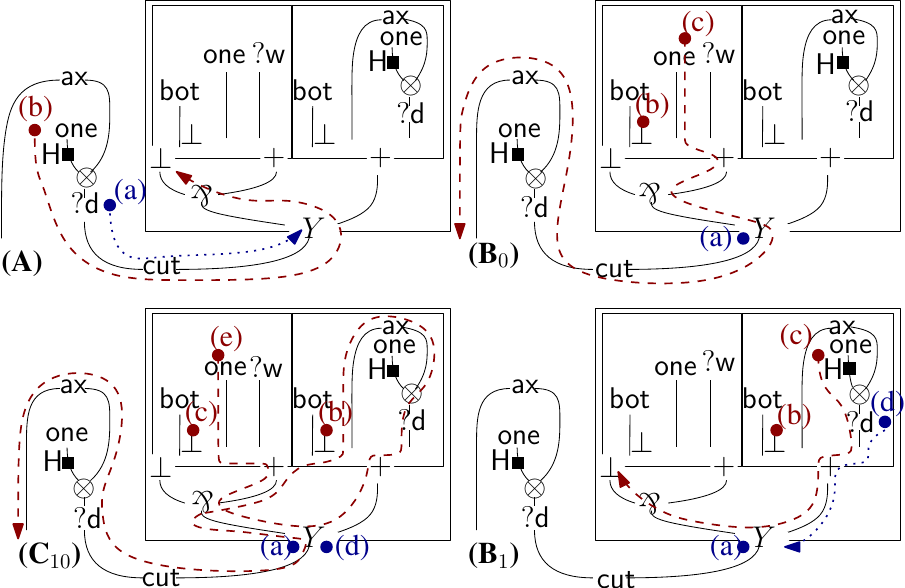}
  \caption{\PSIAM{} run of Fig.~\ref{fig:recursion-net-1}.}
  \label{fig:recursion-msiam}
\end{figure}

In the next sections, we study the properties of the machine, and show that the  \PSIAM\ is a computational model for $\mathcal{N}$.

\paragraph{Example.}
We informally  develop in Fig.~\ref{fig:recursion-msiam} an execution of the
\PSIAM{} for the LHS net of Fig.~\ref{fig:recursion-net-1}. In the
first panel (A) tokens (a) and (b) are generated. Token (a) reaches
the principal door  of the $Y$-box, which corresponds to \emph{opening}  a first copy. Token (b)   enters
the $Y$-box and hits the \bbox.
  The test action of the memory triggers a probabilistic
  distribution of states where the left and the right components of the
  $\bot$-box are opened: the corresponding sequences of operations are
  Panels (B${}_0$) and (B${}_1$) for the left and right sides.

  In Panel (B${}_0$): the left-side of the \bbox\  is opened and its
  $\mathsf{one}$-node emits the token (c) that eventually reaches the
  conclusion of the net.
  In Panel (B${}_1$): the right-side of the \bbox\ is opened and
  tokens (c) and (d) are emitted. Token (d) opens a new copy of the
  $Y$-box, while token (c) hits the \bbox\ of this second
  copy. The test action of the memory again spawns a probabilistic
  distribution.
  
  We focus on panel (C${}_{10}$) on the case of the opening of the
  left-side of the $\bot$-box: there, a new token (e) is generated. It
  will exit the second copy of the $Y$-box, go through the first copy
  and exit to the conclusion of the net.

\subsection{\PSIAM\ Properties, and Deadlock-Freeness}

Intuitively, a \emph{run} of the machine $\M_{\R}$ is the result of
composing transitions of $\M_{\R}$, starting from the initial state
$\stII_\R$ (composition being transitive composition). We are not
interested in the actual order in which the transitions are performed
in the various components of a distribution of states. Instead, we are
interested in knowing which  distributions of states  are reached
from the initial state. This notion is captured well by the
relation $\reach$ (see Sec.~\ref{sect:distrpars}). We will say that \emph{a run of the machine
  $\M_{\R}$ reaches $\mu\in \dists{\sts_\R}$ if $\stII_\R\reach \mu$.}
\condinc{}{ We will also use the expression ``a run of $\M_\R$ reaches a  state $\st$'' if  $\stII\reach \mu$ with $T\in \supp{\mu}$.}

An analysis similar to the one done for program nets shows 
the next lemma (Lem.~\ref{lem:msiamDiamond}) and therefore Prop.~\ref{prop:machine_conf}:
\begin{lemma}[Diamond]\label{lem:msiamDiamond}
  The relation $\redsiam$ satisfies the diamond property.\qed
\end{lemma}
\begin{proposition}[Confluence, Uniqueness of Normal Forms, Uniformity]\label{lem:diamProp}\label{prop:machine_conf} 
  The relation $\redsiam$ satisfies confluence, uniformity, and
  uniqueness of normal forms.\qed
\end{proposition}
By the results we have studied in Sec. \ref{sect:confluence}, we
thus conclude that all runs of $\M_\R$ have the same behavior with respect
to the degree of termination, \ie\ if $\stII_\R$  p-normalizes following a
certain sequence of reductions, it will do so 
whatever sequence of reductions we pick.  
\condinc{}{We say that \emph{the machine $\M_\R$ $p$-terminates} if $\stII_\R$ $p$-terminates (see Def.~\ref{def:normalization}).}

\paragraph{Deadlocks.} A terminal state $\st \stopsiam$ of $\M_\R$ can be final
or not.  A non-final terminal state is called a
\emph{deadlocked} state. Because of the inter-dependencies among
tokens given by the multi-token conditions, a multi-token machine is
likely to have deadlocks.  We are however able to guarantee that any
\PSIAM\ machine is deadlock-free, whatever is the choice for the
memory structure.

\begin{theorem}[Deadlock-Freeness of the \PSIAM]\label{deadlock-free}
  Let $\R$ be a program net of conclusion $\one$; if
  $\stII_\R \reach \mu$ and $\st\in \supp{\mu}$ is terminal, then
  $\st$ is a \emph{final} state.\qed
\end{theorem}

\condinc{}{The proof \condinc{}{(see Appendix~\ref{sec:interplay})} relies on  the diamond
property of the machine (more precisely, uniqueness of the normal forms).}


\subsection{Invariance and Adequacy}\label{semantics}

The machine $\machine{\R}$ gives a computational semantics to
$\R$. The semantics is invariant under reduction
(Th.~\ref{invariance}); the adequacy result (Th.~\ref{adequacy
  psiam}) relates convergence of the machine and convergence of the
nets. We define the convergence of the machine as the convergence of
its initial state:
\begin{center}
$\M_\R\Downarrow_p$ if $\stII_\R$ \emph{converges with probability $p$} (\ie\ $\stII_\R  \Downarrow_p$).
\end{center}

\begin{theorem}[Invariance]\label{invariance} 
  Let $\R$ be a program net of conclusion $\one$.  Assume
  $\R \rednet \sum_i {p_i}\cdot \{\R_i\}$. Then we have that
  $\M_\R\Downarrow_q$ if and only if $\M_{\R_i}\Downarrow_{q_i}$ with
  $\sum_i (p_i\cdot q_i)=q$.\qed
\end{theorem}
\begin{theorem}[Adequacy]\label{thm:termination}\label{adequacy psiam}
  Let $\R$ be a program net of conclusion $\one$.  Then,
  $\M_\R \Downarrow_p$ if and only if $\R \Downarrow_p$.\qed
\end{theorem}
 %
 
\condinc{The proofs of invariance and  adequacy, as well as that of deadlock-freeness, all are based on the \emph{diamond property} of the machine (Lem.~\ref{lem:msiamDiamond}).}{The proofs of invariance and  adequacy, are both based on the \emph{diamond property} of the machine, 
 and on a map---which we call \emph{transformation}---which allows us  to relate the  rewriting of program nets with the \PSIAM{}.
To obtain the proofs,  we need to establish a series of technical results which we give in Appendix ~\ref{app:PSIAM}.}

\section{A {\fontsize{10}{12}\selectfont\PCFwo-style} Language with Memory Structure}\label{sect:PCF}\label{sec:pcfam}
We introduce a \PCFwo-style language which is equipped with a memory
structure, and is therefore parametric on it. The base type will
correspond to elements stored in the memory, and the base operations
to the operations of the memory structure.

\subsection{Syntax and Typing Judgments}

\newcommand{\PCFop}{{\tt c}}
The language \PCF{} which we propose is based on Linear Logic, and is
parameterized by a choice of a memory structure $\memories$.

The \emph{terms} ($ M,N,P$) and \emph{types} ($A,B$) are defined as
follows:
\[
\begin{array}{@{}l@{~}l@{~}l@{}}
  M,N,P
  &{:}{:}{=} &
  x \bor
  \lambda x.M \bor
  MN \bor
  \PCFletp{x,y}{M}{N}\bor
  \PCFpair{M,N} \bor
  \\ &&
  \PCFletrec{f}{x}{M}{N}\bor
  \\ &&
  \PCFnew \bor \PCFop \bor
  \PCFif{P}{M}{N},
  \\
  A,B
  &{:}{:}{=} &
  \PCFalpha \bor
  A \PCFarrow B \bor
  A \PCFprod B \bor \PCFbang A
\end{array}
\]
where $\PCFop$ ranges over the set of memory operations
$\mathcal{L}$. A typing context $\Delta$ is a (finite) set of typed
variables $\{x_1:A_1,\dots,x_n:A_n\}$, and a typing judgment is
written as $ \Delta\PCFentail M:A.  $ An empty typing context is
denoted by ``$\cdot$''. We say that a type is {\em linear} if it is
not of the form $\PCFbang A$. We denote by $\PCFbang\Delta$ a typing
context with only non-linearly typed variables. A typing judgment is
\emph{valid} if it can be derived from the set of typing rules
presented in Fig.~\ref{fig:typrules}.  We require $M$ and $N$ to have
empty context in the typing rule of $\PCFif{P}{M}{N}$.  The
requirement does not reduce expressivity as typing contexts can always
be lambda-abstracted.
The typing rules make use of a notion of \emph{values}, defined as follows:
$U,V~{:}{:}{=}~ x \bor \lambda x.M \bor \PCFpair{U,V} \bor \PCFop$.

\begin{figure*}
{\begin{minipage}{.97\textwidth}
\[
\def\nl{\\[1ex]}
\begin{array}{c}
\infer{\PCFbang\Delta\PCFentail \PCFnew:\PCFalpha}{}
\quad
\infer{\PCFbang\Delta,x:\PCFbang(A\PCFarrow B)\PCFentail x:A\PCFarrow B}{}
\quad
\infer{\PCFbang\Delta,x:A\PCFentail x:A}{A\textrm{ linear}}
\quad
\infer{\PCFbang\Delta\PCFentail V:\PCFbang(A\PCFarrow
  B)}{\PCFbang\Delta\PCFentail V:A\PCFarrow B & V\textrm{ value}}
\quad
\infer{\Delta\PCFentail \lambda x.M:A\PCFarrow B}{\Delta,x:A\PCFentail
  M:B}
\nl
\infer{\PCFbang\Delta,\Gamma_1,\Gamma_2\PCFentail MN:B}{
  \PCFbang\Delta,\Gamma_1\PCFentail M:A\PCFarrow B
  &
  \PCFbang\Delta,\Gamma_2\PCFentail N:A
}
\,\,
\infer{\PCFbang\Delta,\Gamma_1,\Gamma_2\PCFentail \PCFletp{x,y}{M}{N}:C}{
  \PCFbang\Delta,\Gamma_1\PCFentail M:A\PCFprod B
  &
  \PCFbang\Delta,\Gamma_2,x:A,y:B\PCFentail N:C
}
\,\,\,
\infer{\PCFbang\Delta,\Gamma_1,\Gamma_2\PCFentail \PCFpair{M,N}:A\PCFprod B}{
  \PCFbang\Delta,\Gamma_1\PCFentail M:A
  &
  \PCFbang\Delta,\Gamma_2\PCFentail N:B
}
\nl
\infer{\Delta\PCFentail \PCFif{P}{M}{N}:A}{
  \Delta\PCFentail P:\PCFalpha
  &
  \cdot\PCFentail M:A
  &
  \cdot\PCFentail N:A
}
\quad
\infer{\PCFbang\Delta\PCFentail \PCFop : \alpha^{\PCFprod n}\PCFarrow\alpha^{\PCFprod n}}{
    \Marity(\PCFop) = n
}
\quad
\infer{\PCFbang\Delta,\Gamma\PCFentail 
    \PCFletrec{f}{x}{M}{N}:C}{
  \PCFbang\Delta,f:\PCFbang(A{\PCFarrow} B),x:A\PCFentail M:B
  &
  \PCFbang\Delta,\Gamma,f:\PCFbang(A{\PCFarrow}B)\PCFentail N:C
}
\end{array}
\]
\end{minipage}}
\caption{Typing Rules.}
\label{fig:typrules}
\end{figure*}

\subsection{Operational Semantics}
The operational semantics for \PCF{} is similar to  the one of
\cite{PaganiSV14}, and is inherently call-by-value. Indeed, being
based on Linear Logic, the language only allows the duplication of
``!''-boxes, that is, normal forms of ``!''-type: these are the
values. The operational semantics is in the form of a PARS, written
$\PCFcbv$.

The PARS is defined using a notion of {\em 
  reduction context $C[-]$}, defined by the grammar
\[
\begin{array}{@{}l@{~~}l@{~~}l@{}}
  C[-]
  &{:}{:}{=} &
  [-]\bor C[-]N\bor VC[-]\bor \PCFpair{C[-],N}
  \bor \PCFpair{V,C[-]}\\
  &&\bor\PCFletp{x,y}{C[-]}{N}\bor
  \PCFif{C[-]}{M}{N},
\end{array}
\]
and a notion of abstract machine: the \PCFAM{}.
A \emph{raw \PCFAM{} closure} is a tuple
$
(M,\ind{M},\mm)
$
where $M$ is a term, $\ind{M}$ is an injective map from the set
of free variables of  $M$ to $I$, and  $\mm\in\memories$.
\emph{\PCFAM{} closures} are defined as equivalence classes of raw \PCFAM{} closures over permutations of addresses.

\begin{figure*}
  \[
  \begin{array}{c}
    (C[\PCFnew], \ind{}, \mm)
    \PCFcbv_{\mathrm{link}}
    (C[x], \ind{}\cup\{x \mapsto i\}, \mm) 
    \quad
    (C[\PCFop\,\PCFpair{x_1,\ldots,x_n}], \ind{}, \mm) 
    \PCFcbv_{\mathrm{update}(\PCFop)}
    (C[\PCFpair{\vec{x}}],\ind{},
    \Mupdate(\vec{i},\PCFop,\mm))
    \\[1ex]
    (C[\PCFif{x}{M_{\PCFtt}}{M_{\PCFff}}], \ind{}, \mm)
    \PCFcbv_{\mathrm{test}(i)}
    \Mtest(i,\mm)[\Mtrue := (M_{\PCFtt},\ind{}\setminus\{x \mapsto i\}), \Mfalse
    := (M_{\PCFff},\ind{}\setminus\{x \mapsto i\})]
    \\[1ex]
    (C[(\lambda x.M)U], \ind{}, \mm)\PCFcbv
    (C[M\{x := U\}], \ind{}, \mm)
    \,\,\,
    (C[\PCFletp{x,y}{\PCFpair{U,V}}{M}],\ind{},\mm) \PCFcbv
    (C[N[x := U,y := V]],\ind{},\mm)
    \\[1ex]
    (C[\PCFletrec{f}{x}{M}{N}],\ind{},\mm)
    \PCFcbv
    (C[N\{f := \lambda x.\PCFletrec{f}{x}{M}{M}\}],\ind{},\mm)
  \end{array}
  \]
  \caption{Rewrite System for \PCFAM.}
  \label{tab:rw-pcf}
\end{figure*}

The rewrite system is defined in Fig.~\ref{tab:rw-pcf}.
First, the creation of a new base type element ($\PCFcbv_{\mathrm{link}}$) is
simply memory allocation: $x$ is fresh (and not bound) in $C$ and $i$
is a new address neither in the image of $\ind{}$ nor in the support
of $\mm$.
Then, the operation $\PCFop$ reduces through $\PCFcbv_{\mathrm{update}(\PCFop)}$ using
the update of the memory when $\Marity(\PCFop) = n$ and
$\ind{}(x_k) = i_k$.
Then, the if-then-else reduces through $\PCFcbv_{\mathrm{test}(i)}$ using the test
operation where $\ind{}(x) = i$. Note how we remove $x$ from the
domain of $\ind{}$.
Finally we have the three rules that do not involve probabilities:
Note how the mapping $\ind{}$ can be kept the same: the set of free
variables is unchanged.

Let $\PCFam = [(M,\ind{}, \mm)]$ be a \PCFAM{} closure. We define the judgment
$
x_1:A_1, \ldots, x_m:A_m
\PCFentail
\PCFam : B
$
if none of the $x_i$'s belongs to ${\rm Dom}(\ind{})$,
$
y_1:\PCFalpha,\ldots,y_k:\PCFalpha,x_1:A_1,\ldots,x_m:A_m
\PCFentail M : B
$,
and $\{y_1,\ldots,y_k\} = {\rm Dom}(\ind{})$.

\subsection{Modeling \PCF\ with Nets}
We now encode \PCF{} typing judgments and typed \PCFAM{} closures into program
nets.
As the type system is built on top of Linear Logic, the translation
$\net{(-)}$ is rather straightforward, modulo one subtlety: it is
parameterized by a memory structure $\mm$ and a partial function
$\ind{}$ mapping term variables to addresses in $I$.

The mapping $(-)^\dagger$ of types to formulas is defined by
$\PCFalpha^\dagger ~{:}{=}~\one$,
$(A\PCFarrow B)^\dagger ~{:}{=}~ {(A^\dagger\b\parr B^\dagger)}$ and
$(A\PCFprod B)^\dagger ~{:}{=}~ {A^\dagger}\tens{B^\dagger}$.
Now, assume that $\{y_1,\ldots,y_n\}\cap{\rm Dom}(\ind{})=\emptyset$,
that $\Delta$ is a judgment whose variables are all of type
$\PCFalpha$, and that $|\Delta|={\rm Dom}(\ind{})$. The typing
judgment $y_1:A_1,\ldots,y_n:A_n,\Delta\PCFentail M:A$ is 
mapped through $\net{(-)}_{\ind{},\mm}$ to a program net
$\net{M}_{\ind{},\mm}=[(R_M,\ind{R_M}, \mm)]$ with conclusions
$(A_1^\dagger)^{\bot}$,\ldots $(A_n^\dagger)^{\bot}$, $(B^\dagger)$
and memory state $\mm$ (note how the variables in $\Delta$ do not
appear as conclusions). 
\condinc{}{The full definition is found in
Appendix~\ref{app:PCFproofs}.}

\subsubsection{Adequacy}

As in Sec.~\ref{sect:confluence} and~\ref{semantics}, given a \PCFAM{} closure $\PCFam$ we
write $\PCFam\Downarrow_p$ ($\PCFam$ converges to $p$) if
$p = \mathrm{sup}_{\PCFam \PCFreddcbvx \mu} \NF{\mu}$.
The adequacy theorem then relates convergence of programs and
convergence of nets. 
\condinc{}{A sketch of the proof is given in Appendix~\ref{app:PCFproofs}.}

\begin{theorem}\label{th:adeqcbv}
  Let $\PCFentail \PCFam:\PCFalpha$, then $\PCFam\,{\Downarrow_p}$
  if and only if $\PCFsemcbv{\PCFam}{\Downarrow_p}$.\qed
\end{theorem}

\section{Results and Discussion}

As we  anticipated in Sec.~\ref{overview}, we have proved---{\em parametrically} on the {\em memory}---that 
the \PSIAM\ is an adequate model of program nets reduction (Th.~\ref{adequacy psiam}),  
and program nets are expressive enough  to adequately  represent the
behavior of the \PCF\ abstract machine (Th.~\ref{th:adeqcbv}).
What does this mean? As soon as we choose a concrete instance of memory structure, we have a language and an adequacy result for it. This is in particular the case for all instances of memory which are outlined in Sec.~\ref{sec:exemples}. 
%
To make this explicit, let  $\Imem$, $\Pmem$ and  $\Qmem$ be respectively a deterministic, probabilistic and quantum memory. We denote  by  \PCF($\Imem$), \PCF($\Pmem$) and \PCF($\Qmem$), respectively,  the language which is obtained by choosing that memory.
Observe in particular that the choice of $\Pmem$ or $\Qmem$, respectively specialize our adequacy result  into  a semantics  for a probabilistic \PCFwo\  in the style of \cite{EhrhardTassonPagani}, and a semantics for a quantum \PCFwo, 
in the style of \cite{SelingerValiron,PaganiSV14}.

\subsection{The Quantum Lambda Calculus}\label{sec:quantum literature}
Let us now focus on the quantum case, and analyze in some depth our result.
We  have a quantum lambda-calculus, namely  \PCF($\Qmem$),  together
with an adequate multi-token semantics.  How does  our calculus relate with the ones in  the literature?

We first observe that the {\em syntax} of \PCF($\Qmem$) is very close
  to the language of \cite{PaganiSV14} (we only omit lists and
  coproducts). The {\em operational semantics} is also the same, as
  one can easily see.
  Indeed, the abstract machine in~\cite{PaganiSV14} consists of a triple
  $(Q,L,M)$ where $M$ is a lambda-term and where $Q$ and $L$ are as
  presented in Sec.~\ref{sec:quantum-memory}. As we discussed
  there, for $Q$ and $L$ one can use either the canonical presentation
  of~\cite{PaganiSV14}, or the memory structure $\Qmem$.

\subsubsection{Discussion on the Quantum Model}\label{sec:discussion}

It is now time to go back to the programs in our motivating examples,
Examples~\ref{ex:entangled} and~\ref{ex:recursion}.  Both programs
are valid terms in \PCF($\Qmem$); we have already  informally developed
Example~\ref{ex:recursion} within our model.

We claimed in the Introduction that Example~\ref{ex:entangled} cannot
be represented in the GoI model described in \cite{hasuo11}: the reason is that the
model does not support entangled qubits in the type
$\PCFalpha\otimes\PCFalpha$ (using our notation), a tensor product is
always separable. To handle entangled states, \cite{hasuo11} uses
non-splittable, crafted types: this is why the simple term in Example~\ref{ex:entangled}
is forbidden. In the \PSIAM{}, entangled states pose no problem, as the memory is disconnected from the types.
  
The term of Example~\ref{ex:recursion}, valid in \PCF($\Qmem$), is mapped
through $(-)^\dagger$ to the net of Fig.~\ref{fig:recursion-net-1}:
Th.~\ref{th:adeqcbv} and~\ref{thm:termination} state that the
corresponding \PSIAM{} presented in Fig.~\ref{fig:recursion-msiam}
is adequate. Note that Example~\ref{ex:recursion} was presented in the context of
quantum computation. It is however possible (and the behavior is going
to be the same as the one already described) to use the probabilistic
memory sketched in  Sec.~\ref{sec:prob-bool}. In this case, the
$\mathsf{H}$-sync node would be changed for the coin-sync node.

\subsubsection{Qubits, Duplication and Erasing}\label{why surface} It is worth  to pinpoint
the technical ingredients which allow for the coexistence of quantum
bits with duplication and erasing.  In the language, the reason is
that, similarly to \cite{PaganiSV14}, \PCF{} allows only
lambda-abstractions (or tuples thereof) to be duplicated.  In the
case of the nets (and therefore  of the \PSIAM), the key ingredient is
\emph{surface reduction} (Sec.~\ref{sec:nets}): the allocation of a quantum bit is
captured by the \textit{link} rule which associates a $\onelk$ node
to the memory. Since a  $\onelk$ node linked to the
memory \emph{cannot} lie inside a box, it will \emph{never be copied nor erased}.  Indeed, the 
ways in which the language and the model deal with quantum bits,
actually match.

\condinc{}{ 

\subsubsection{More Possibilities for Quantum Memory Structure}

In the presentation of $\Qmem$ we have given in
Sec.~\ref{sec:quantum-memory}, fresh qubits were arbitrarily
created in state $\ket0$. We could as well have chosen, say, $\ket1$
for default value.
Formally, if we want to do that we can use the set $\mathcal F_1$ of
\mbox{(set-)maps} from $I$ to $\{0,1\}$ that have value $1$ everywhere
except for a finite subset of $I$. The structure $\Qmem$ would
then have been defined as $\mathcal H_1$, the Hilbert space built from
finite linear combinations of $\mathcal F_1$.

Unlike the case of the integer memory, the mathematical properties of
quantum states can make the test action modify the state of the fresh
addresses. Let us see how this happens.

Indeed, one can not only build a memory structure $\Qmem$ with
$\mathcal H_0$ and $\mathcal H_1$ but also with a superposition of the
elements in $\mathcal H_0$ and $\mathcal H_1$.  For example, one could
choose
$\Qmem = \{\alpha{}v_0+\beta{}v_1) | v_x\in\mathcal H_x,
|\alpha|^2+|\beta|^2=1\}$.
This makes a memory structure satisfying all the equations. In this
system, a valid memory can have all of its fresh variables in
superposition. Any measurement on a fresh variable of this memory will
collapse the state... and ``modify'' the global state of the fresh
variables. So, despite the fact that a test on $i$ cannot touch
another address, it can globally act on the memory.
	
This paradox is of course solved when remembering that measurements
and unitary operations (and measurements and measurements) do commute
independently of the state on which they are applied
(Remark~\ref{rem:nonlocal}). So the fact that the memory changes
globally after a test is irrelevant.

} 

\subsection{Conclusion}

In this paper, we have introduced a parallel, multi-token Geometry of
Interaction capturing the choice effects with a parametric memory. This
way, we are able to represent {\em classical, probabilistic and
quantum effects}, and  adequately model the
linearly-typed language \PCF{} parameterized by the same memory
structure. We expect our approach to capture also
non-determin\-istic choice in a natural way: this is ongoing work.

\condinc{}{
\section*{Acknowledgments}
This work has been partially supported by the JSPS-INRIA Bilateral Joint Research Project CRECOGI.
%
%
%
A.Y.\ is supported by Grant-in-Aid for JSPS Fellows.
}

\newpage
\setlength{\bibhang}{250pt}
\bibliographystyle{unsrtnat}
\bibliography{biblio}

\clearpage
\appendix

\section{Commutation of Tests and Updates on Memory States}
\label{app:commut-mem}

The commutation of tests and updates is formally defined as follows.
Assume that $i\neq j$, that $j$ does not meet $\vec{k}$, and that
$\vec{k}$ and $\vec{k'}$ are disjoint.
\begin{itemize}
\item Tests on $i$ commute with tests on $j$. More
    precisely, if 
    \begin{itemize}
    \item 
      $\Mtest(i,m) = p_0\{(\Mtrue,m_0)\} + p_1\{(\Mfalse,m_1)\}$
    \item 
      $\Mtest(j,m_0) = p_{00}\{(\Mtrue,m_{00})\} +
      p_{01}\{(\Mfalse,m_{01})\}$
    \item 
      $\Mtest(j,m_1) = p_{10}\{(\Mtrue,m_{10})\} +
      p_{11}\{(\Mfalse,m_{11})\}$
    \end{itemize}
    and if 
    \begin{itemize}
    \item 
      $\Mtest(j,m) = q_0\{(\Mtrue,m'_0)\} + q_1\{(\Mfalse,m'_1)\}$
    \item 
      $\Mtest(i,m'_0) = q_{00}\{(\Mtrue,m'_{00})\} +
      q_{01}\{(\Mfalse,m'_{01})\}$
    \item 
      $\Mtest(i,m'_1) = q_{10}\{(\Mtrue,m'_{10})\} +
      q_{11}\{(\Mfalse,m'_{11})\}$
    \end{itemize}
    then for all $x,y=0,1$, $m_{xy} = m'_{yx}$ and $p_xp_{xy} = q_yq_{yx}$.
\item Tests of $j$ commute with updates on $\vec k$. More
    precisely, if
    \begin{itemize}
    \item $\Mtest(i,m) = p_0\{(\Mtrue,m_0)\} + p_1\{(\Mfalse,m_1)\}$
    \item $\Mupdate(\vec k, x, m_0) = m_0'$
    \item $\Mupdate(\vec k, x, m_1) = m_1'$
    \end{itemize}
    and if $\Mupdate(\vec k, x, m) = m'$
    then
    \[
      \Mtest(i,m') = p_0\{(\Mtrue,m'_0)\} + p_1\{(\Mfalse,m'_1)\}.
    \]
  \item Updates on $\vec k$ and $\vec k'$ commute. More precisely:
    \begin{multline*}
    \Mupdate(\vec k, x, \Mupdate(\vec k', x', m)) =\\
    \Mupdate(\vec k', x', \Mupdate(\vec k, x, m))
    \end{multline*}
\end{itemize}

\section{Program Nets: proof of the Diamond Property}\label{app:nets}
We prove that the PARS $(\mathcal{N}, \rednet)$ satisfies the diamond property (Prop.\ref{prop:netDiamond} ).
 We write $(R, \ind{R}, \mm) \stackrel{r}{\rednet} \mu$ for
    the reduction of the redex $r$ in the raw program net $(R, \ind{R}, \mm)$.

First, we observe the following property, proven by case analysis.

\begin{lemma}[Locality of $\rednet$]\label{lem:locality}
  Assume that $\R=[(R, \ind{R}, \mm)]$ has two distinct redexes $r_1$
  and $r_2$, with $\R\stackrel{r_1}\rednet\mu_1$,
  $\R\stackrel{r_2}\rednet\mu_2$ and $\mu_1\not=\mu_2$. Then the redex
  $r_2$ (resp. $r_1$) is still a redex in each
  $(R', \ind{R'}, \mm') \in \supp{\mu_1}$ (resp. $\supp{\mu_2}$).\qed
\end{lemma}

The proof of Prop. \ref{prop:netDiamond} goes as follows.

\begin{proof}(of Prop. \ref{prop:netDiamond}.)
  The locality implies the following two facts:
  
  \noindent
  (1) If $(R, \ind{R}, \mm) \rednet \mu$ with
  $\term{\mu} \neq \emptyset$, then the raw program net
  $(R, \ind{R}, \mm)$ contains exactly one redex.

  \noindent
  (2) If $(R, \ind{R}, \mm) \stackrel{r_1}{\rednet} \mu$ and
  $(R, \ind{R}, \mm) \stackrel{r_2}{\rednet} \xi$ with $\mu \neq \xi$,
  then there exists $\rho$ satisfying $\mu \redd \rho$ and
  $\xi \redd \rho$.  Concretely, $\mu \redd \rho$ is obtained by
  reducing the redex $r_2$ in each
  $(R', \ind{R'}, \mm') \in \supp{\mu}$, and
  $\xi \redd \rho$ is obtained by reducing $r_1$.

  Assuming $\mu \redd \nu$ and $\mu \redd \xi$, item 1.\ implies
  $\term{\nu} = \term{\xi}$, and item 2.\ implies
  $\exists \rho. \nu \redd \rho \land \xi \redd \rho$.  Let us review
  some of the non-evident cases explicitly.

  If $r_1$ and $r_2$ are both non-active $\onelk$ nodes, say $x$ and
  $y$ respectively, $(R, \ind{R}, \mm)$ reduces to
  $(R, \ind{R}\cup\{x\mapsto i, y\mapsto j\}, \mm)$ and
  $(R, \ind{R}\cup\{x\mapsto k, y\mapsto l\}, \mm)$ for some fresh
  indexes $i$, $j$, $k$, $l$.  The permutation $(i,k) \circ (j,l)$
  renders the two program nets equivalent.

  If both $r_1$ and $r_2$ modify memories (\ie\ they perform either
  $\Mupdate$ or $\Mtest$), the property holds because the injectivity
  of $\ind{R}$ guarantees that we always have the requirement
  (disjointness of indexes) of the equations given in Appendix~\ref{app:commut-mem}.
  Hence the two reductions commute both on memory
  (up to group action) and on probability.
\end{proof}

\section{\SIAM: Multitoken Conditions, Formally}\label{sec:multi-token}
\paragraph{Stable Tokens.} 
A token in a stable position is said to be \emph{stable}.  Each such
token is the remains of a token which started its journey from $\DER$
or $\ONES$, and flowed in the graph ``looking for a box''.  This
stable token therefore witnesses the fact that \emph{an instance} of
dereliction or of $\onelk$ ``has found its box''.  Stable tokens keep
track of box copies; let us formalize this. Let $S$ be either $R$, or a structure associated to a box (at any depth).
 Given a state $T$ of $\M_R$, we define $\id(S)$ to be $\{\emp\}$ if $R=S$
(we are at depth 0).  Otherwise, if $S$ is the structure associated to
a box node $\BB$ of $R$, we define $\id(S)$ as the set of all $t$ such
that $(e, s, t)$ is a stable token on  the  premiss(es) of $\BB$'s principal door.
  Intuitively, each such $t$ \emph{identifies a copy of the box}
which contains $S$.

\paragraph{Multitoken Conditions: Synchronization, Choice, and Boxes Management.}
Rules marked by (i), (ii), and (iii), (iv) in Fig.~\ref{fig:trRules} 
only apply if the following  conditions are satisfied.
\begin{varitemize}
\item[(i)] Tokens cross a sync node $l$ only if for a certain $\bstk$, there is
  a token on each position $(e,\stk,\bstk)$ where $e$ is a premise of
  $l$, and $\stk$ indicates an occurrence of atom in the type of
  $e$. In this case, all tokens cross the node simultaneously.
  Intuitively, insisting on having the same stack $\bstk$ means that
  the tokens all belong to the same box copy.
\item[(ii)]  
  A token $(e,s,t)$ on one of the conclusions $\Gamma$ of the \bbox\ can move inside
  the box \emph{only if} its box stack $\bstk$ belongs to $\id(S_0)$
  (resp. $\id(S_1)$), where $S_0$ (resp. $S_1$) is the left
  (resp. right) content of the \bbox.  Note that if the \bbox\ is
  inside an exponential box, there could be several stable tokens on each  
  premise of the principal door, one stable token for each copy of the box.
\item[(iii)] 
  The position $\pp=(\edg,\emp,\bstk)$ under  a $\onelk$ node
  (resp. $(\edg,\delta,\bstk)$ under a $\dlk$ node) is added to the
  state $T$ only if: it does not already belong to
  $\orig(T)$, and $\bstk\in \id(S)$, where $S$ is the
  structure to which $e$ belongs.
   If both conditions are satisfied, $T$ is extended with the position $\pp$ (and $\orig(\pp)=\pp$).
   Intuitively, each $(\edg,\emp,\bstk)$ (resp.  $(\edg,\delta,\bstk)$)
   corresponds to a copy of $\onelk$ (resp. $?d$) node.
\end{varitemize}

\section{\PSIAM}  \label{app:PSIAM}


 The proofs of invariance, adequacy, and deadlock-freeness, all are based on the diamond property of the machine, 
 and on a map---which we call \emph{Transformation}---which allows us  to relate the  rewriting of program nets with the \PSIAM{}.
 In this section we establish the technical tools we need. In Sec. \ref{sec:invariance} we prove Invariance, in Sec. \ref{sec:interplay} we proof adequacy and deadlock-freeness.

  The tool we use  to relate net rewriting and the \PSIAM{} is a mapping
 from states of $\R$ to states of $\R_i$,
 which we are going to introduce in this section. 
     This tool together with confluence (due to the diamond property) allows us to establish the  main result of this section,   from which Invariance  (Th.~\ref{invariance}) follows.
  
  \vskip 4pt

  From now on, we use the following conventions and assumption.

\begin{itemize}
\item  The letters $\rstT,\rstU$ range over \emph{raw} \PSIAM\ states, the letters $\stTT,\stUU$ over \PSIAM\ states, and the  letters $T,U$ over \SIAM\ states.

\item  To keep the notation light, we will occasionally rely on our  convention of denoting  the distribution 
   $\{\st^1\}$ by $\{\st\}$ or even simply  by 
      $\st$, when there is no ambiguity.

\item  We assume that $\R\rednet\sum_i{p_i}\cdot\{\R_i\}$,
  where $i\in \{0\}$ or $i\in \{0,1\}$, $\R=[(R,\ind{R},\mem{R})]$,
$\R_i=[(R_i,\ind{R_i},\mem{R_i})]$.

 \item  Unique initial state. We assume that $\R$ has a single conclusion, which has type $1$. As a consequence, $\mathrm{Dom}(\ind{R})=\ONES_{R}$, and   $\machine{\R}$ has  a \emph{a unique raw initial state}, which is $\mathcal I_{\R} =(\emptyset, \emptyset, \mem{R})$. We have $\stII_{\R}=\{\mathcal I_{R}\}$.

 \item \condinc{}{We denote by $\stsI{R}$ the set of states $\st$ which can be reached
from $\stII_\R$, \ie\ $\stII_\R\reach \mu$ and $\st\in \supp{\mu}$.}
 \item we do not insist too much on the distinction between raw states and states, which in this section is not relevant.

  \end{itemize}
\subsection{Properties and Tools}

In this section, most of the time we analyze the reduction of raw program nets and raw states, because we do not need to use the equivalence relation. Which is the same: we pick a representative of the class, and follow it through its reductions. 

\subsubsection{Exploit the Diamond}\label{sect:construction}
Because the \PSIAM\ is diamond, we can always pick a run of the machine which is convenient for us to analyze the machine. By confluence and uniqueness of normal forms, all choices produce the same result w.r.t. both the degree of termination of any distribution which can be reached (invariance), and the  states which are reached (deadlock-freeness). 

In case of $\R\rednet \rho$ via $\Mlink, \Mupdate$ or $\Mtest$, we will always choose a run which begins as indicated below:

\begin{enumerate}

\item \textit{Link.}  Assume $(R,\ind{R},\mem{R})
  \rednet_{\mathrm{link(n,j)}} \{(R, \ind{R}\cup\{x\mapsto j\},\mem{R})\}$.  The machine does the same: from the initial state  the machine transitions using its reduction $\Mlink(n,j)$, on the same $\onelk$ node. We can choose the same address $j$ because we know it is fresh for $\mem{R}$. Therefore we have $(I, \ind{I}, \mem{R}) \red_{\Mlink(n,j)} \{(\stU, \ind{U},\mem{R})\}=\mu$.

\item \textit{Update.} Assume $(R,\ind{R},\mem{R}) \rednet_{\Mupdate(s)}$

\noindent 
$ \{(R', \ind{R},  \upd(l,\vec i,\mem{R})\}$. Observe that the $\onelk$ node $n$ in the redex is active; let $j$ be the corresponding address. We choose a run which starts with the transitions $(I, \ind{I},\mem{R})$

\noindent
$\red_{\Mlink(n,j)}\{ (\stU, \ind{U},\mem{R})\}$ and $(\stU, \ind{U},\mem{R}) \red_{\Mupdate(s)}$\\
$ \{(\stU, \ind{U},  \upd(l,\vec i,\mem{R})\}=\mu   $.

\item \textit{Test.} Assume  $(R,\ind{R},\mem{R}) \rednet_{\Mtest(j)} \rho$  where for each $i$,  $R\rednet_{u_i} R_i $. Again  the $\onelk$ node $n$ in the redex is active; let $j$ be the corresponding address. Our canonical way to start the run of the machine applies to the initial state  $(I, \ind{I},\mem{R})$ the transition $\Mlink(n,j)$,
crosses the cut, and finally applies the same $\Mtest(j)$, to reach  $\Mtest(j,\mem{R})[\Mtrue:=\stU_0,\Mfalse:= \stU_1])=\mu$.
\end{enumerate}

  \subsubsection{The Transformation Map}

  The tool we use  to relate net rewriting and the \PSIAM{} is a mapping
  from  states of $\M_\R$ to states of $\M_{\R_i}$. 
    We first define a map on positions of $R$, then on \SIAM\ states, and finally on \PSIAM\ raw states.

  \begin{figure*}[htbp]
        \begin{center}
              \includegraphics[width=10cm]{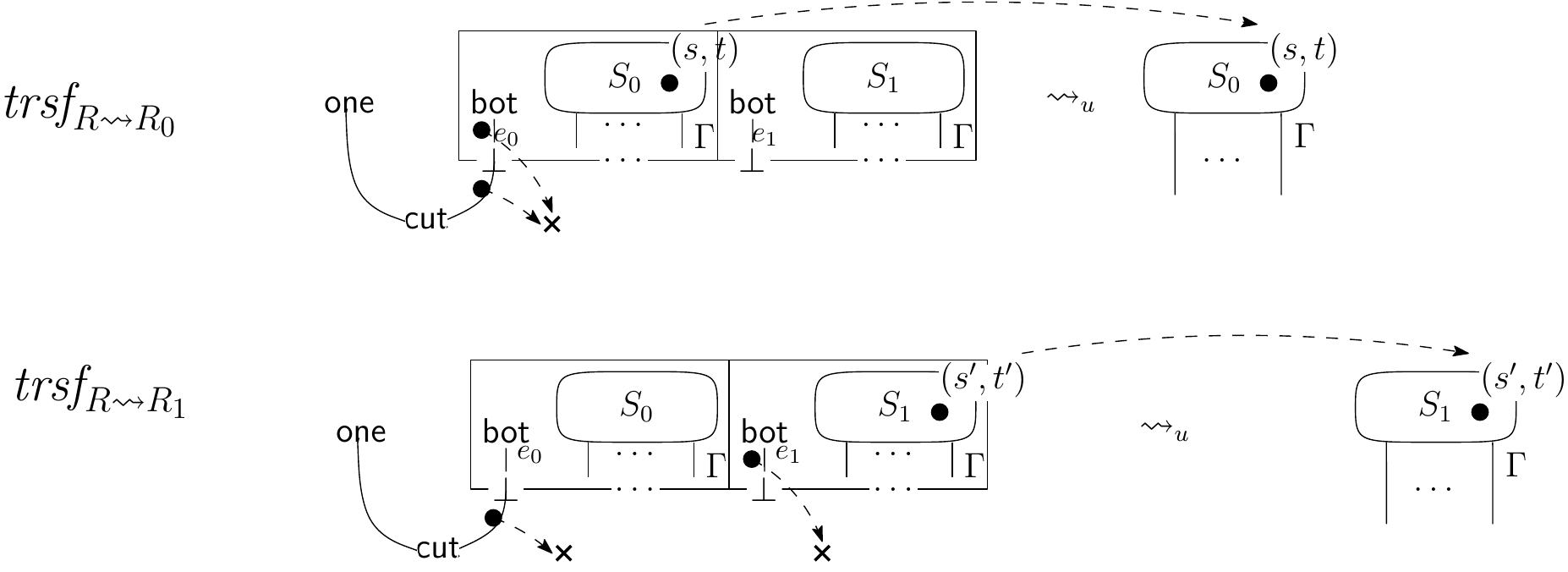}
        \end{center}
        \caption{The Partial Function $\trsf$ on \bbox\ Reduction.}\label{fig:trsfBox}
      \end{figure*}

     %

 \paragraph{Transformation of \SIAM\ States.}    
For each  $R_i$ to which $R$ reduces, we define a \emph{transformation}  on positions, as  a partial function 
  $\trsf_{R \rednet R_i}: \POSALL_{R} \partto \POSALL_{R_i}$. The key case is the case of \bbox\ reduction, illustrated in Fig. \ref{fig:trsfBox};  for each position  outside the redex, we intend that  $\trsf (\pp)$   is the identity.
  The other cases are as in \cite{lics2015}.
      \begin{figure*}[htbp]
        \begin{center}
              \includegraphics[width=15cm]{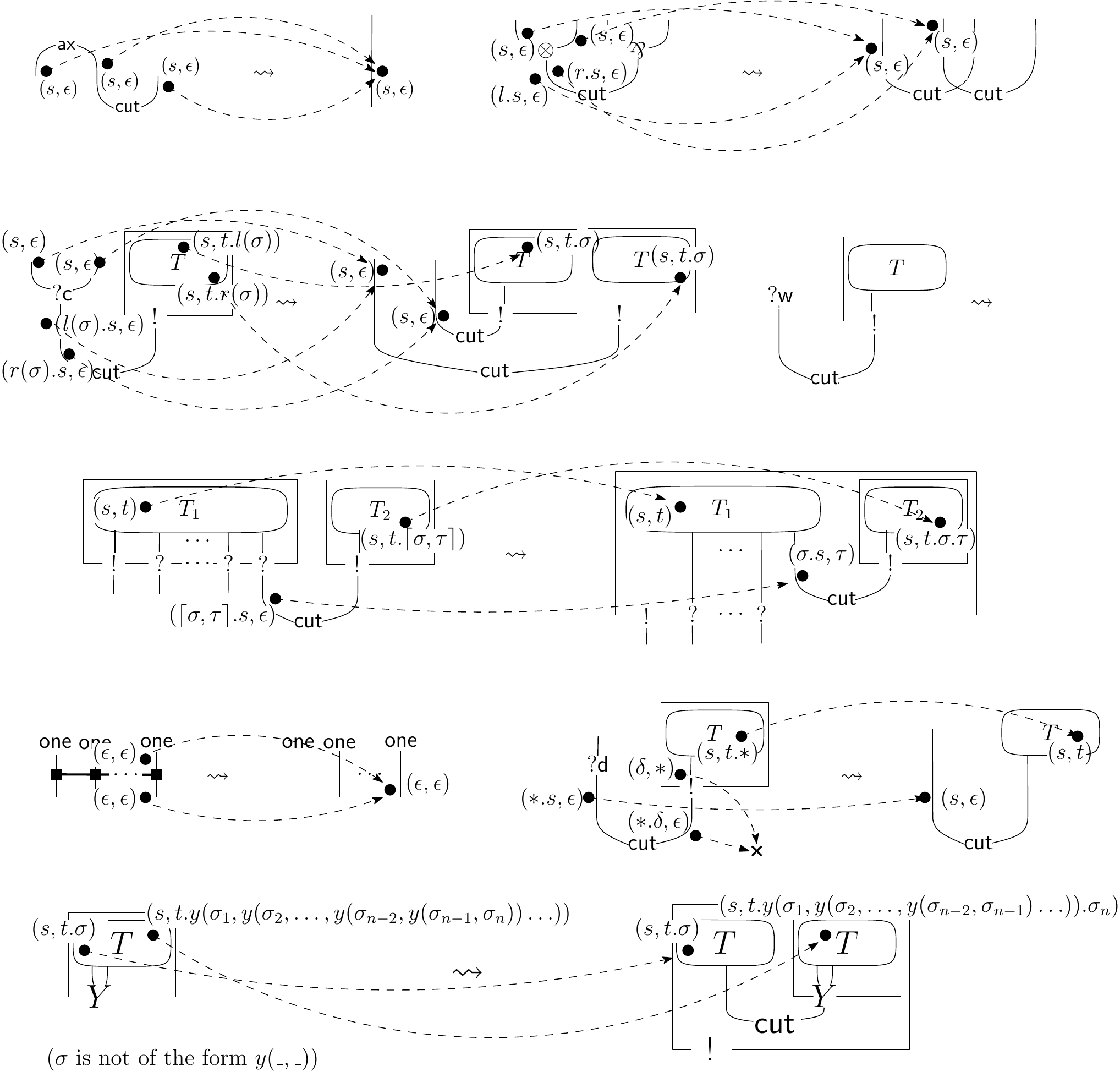}
        \end{center}
        \caption{The Function $\trsf_{R\rednet R'}$.}\label{fig:trsfOthers}
      \end{figure*}
      
The definition extends to the states of the \SIAM{}
point-wisely, in the obvious way.

From now on, we write $\trsf_{R_i}$  or sometimes simply  $\trsf_i$
for $\trsf_{R\rednet R_i}$.

\paragraph{Transformation of \PSIAM\ States.} We now extend   $\trsf_{R\rednet R_i}$ to \PSIAM\  states. 
To do so smoothly, we define a subset  $[\trsf_{R_i}]$ of $\stsI{R}$, which depends on the reduction rule. To work with such states simplify the proofs, and  is always possible because of  Sec. \ref{sect:construction}.

\begin{itemize}

\item Case $\rednet_{\Mlink(n,j)}$.  We define $[\trsf_{R_i}]$ as the set of the states in $\stsI{R}$ in which $\ind{}(\pp)=j$, where  $\pp$ is the position associated to the $\onelk$ node $n$.
 
 \item  Case $\rednet_{\Mupdate(s)}$. We define $[\trsf_{R_i}]$ as the set of the states in $\stsI{R}$ which ``have crossed'' the sync node $s$. We can easily characterize these states. 
 Assume  $ \pp_1, \ldots, \pp_n$ are the positions associated to the premises of $s$ (observe that each $\pp_i$  belongs to $\ONES_R$). $[\trsf_{R_i}]$ is the set of the states $\stTT \in \stsI{R}$ such that $\{\pp_1, \ldots, \pp_n\}\subseteq  \orig(T)$ and      $\{\pp_1, \ldots, \pp_n\}\not \subseteq T$.

\item  Case  $\rednet_{\Mtest(j)}$.  We define  $[\trsf_{R_0}]$ as the set of the states in $\stsI{R}$ which have a 
 token on the left $\botlk$ of the redex (the edge $e_0$ in the Fig.~\ref{fig:trsfBox}).  We define $[\trsf_{R_1}]$ similarly.

\item  Otherwise: we define $[\trsf_{R_i}] = \stsI{R}$.
 \end{itemize}

\begin{definition}[Transformation Map]
~
\begin{enumerate}
\item  $\trsf_i:[\trsf_{R_i}]\subset \stsI{R} \to \sts_{R_i}$  maps the  state $\stTT=[(T, \ind{}, \mem{})]$ 
 into $[\trsf_i(T, \ind{}, \mem{})]$, with  
 $$\trsf_i(T, \ind{}, \mem{})= (\trsf_i (T), \ind{}, \mem{} ).$$ 
\item The definition  extends linearly to  \emph{distributions}. 
     Assume  $\mu=\sum c_k \cdot \{\st_k\}$ 
     and $ \st_k\in [\trsf_{R_i}] $  for each $\st_k$,
      then 
      $$\trsf_{R \rednet R_i} (\mu):=\sum c_k \cdot \{\trsf_{R \rednet R_i}(\st_k)\}.$$
\end{enumerate}

\end{definition}

\begin{fact} If $\st\in [\trsf_{R_i}]$, with $\st \red \mu$ and $\sttwo\in \supp{\mu}$, then $\sttwo\in [trsf_{R_i}]$.
\end{fact}
 
 \begin{lemma}[Important Observation]\label{lem:main}
   The construction given in \ref{sect:construction} leads each time to a distribution $\mu$, where  each state in the support
      satisfies:
      
             \begin{itemize}
                  \item  $\stUU_i\in [\trsf_{R_i}].$
             \item  $  \trsf_{R_i}(\stUU_i) =\stII_{R_i}$.

             \end{itemize}

 \end{lemma}

\subsubsection{Properties of the Reachable States}
  \newcommand{\future}[1]{\mathop{\uparrow}(#1)}
Let us analyze the set of states which is  spanned by a run of the \PSIAM.
Given a \bbox\ of $\R$, let $e_0$ be the conclusion of the left $\bot$ and $e_1$ be the conclusion of the right $\bot$.
For any stacks $s,t$, we call the two stable positions $(e_0,s,t)$ and $(e_1,s,t)$ a  $\bot$-pair. These two positions are mutually exclusive in a state, because $\orig(e_0,s,t) = \orig(e_1,s,t)$.

 We say that two states $\st,\sttwo\in \stsI{R}$ are \emph{in conflict}, written $\st \smile \sttwo$, if $\st$ contains one of the two positions of a $\bot$-pair and $S$ the other.
 We observe that conflict is hereditary with respect to transitions, because stable positions are never deleted or modified by a transition.
 Let $\future{\st} = \{
 \sttwo ~|~ \st\mathrel{\redd^*}\rho \,\land\, \sttwo\in \supp{\rho}
 \}$. The following properties are all immediate:
 
 \begin{enumerate}

 \item If $\st \smile \st'$, $\sttwo \in \future{\st}$, and $\sttwo' \in \future{\st'}$, then $\sttwo \smile \sttwo'$.
 \item If $\st \red\mu$, either the transition is deterministic, or  $\supp{\mu}= \{\sttwo_0,\sttwo_1\}$ with $\sttwo_0 \smile \sttwo_1$.
 \item If $\stII\redd^* \mu$, then for each $\st\not=\st'\in \supp{\mu}$,  $\st \smile \st'$.

 \end{enumerate}

States in conflict are in particular disjoint. Therefore we can safely sum them:

\begin{lemma}\label{lem:conflict}
Given a distribution of states $\mu \in \dists{\stsI{R}}$,
  \[
  \infer
      { \mu \redd^{k}   \sum_{\st\in\supp{\mu}}\mu(\st)\cdot\rho_{\st}  }
      {  \forall \st_i,\st_j\in \supp{\mu}.  \st_i \smile \st_j~   &  \{\st \redd^k \rho_{\st}\}_{\st\in\supp{\mu}}  }
      \,.
  \]
\end{lemma}
              
As an immediate consequence, the following also hold:

\[
\infer{ \stUU \redd^{k+1}   \sum_{\st\in\supp{\mu}}\mu(\st)\cdot\rho_{\st}  }
      { \stUU \red \mu    &  \{\st \redd^k \rho_{\st}\}_{\st\in\supp{\mu}}  }
      \]
\[
      \infer{ \stUU \redd^{n+k}   \sum_{\st\in\supp{\mu}}\mu(\st)\cdot\rho_{\st}  }
            { \stUU \redd^n \mu    &  \{\st \redd^k \rho_{\st}\}_{\st\in\supp{\mu}}  }
\]

   \subsubsection{The Reachability Relation $\reach$}    
   The reachability relation $\reach$ (defined  in  Sec.~\ref{sect:distrpars}) 
   is a useful tool in the study of the \PSIAM. 
   
   %
   A \emph{derivation} of $a\reach \mu$ is inductively obtained by using the  rules which define $\reach$.
   %
   %

  In the case of the \PSIAM\, the relations $\reach$
   and $\redd$ are equivalent with respect to normal forms.

   \begin{lemma}\label{lem:equiv}
     If $\{\st\} \redd^n \xi $ then $\st \reach \xi$. 
   Conversely, if  $\st \reach \mu$ then there exists $\rho$ with   $\{\st\} \redd^*
   \rho$ and such that $\term \mu \subseteq \term \rho$.
   \qed
   \end{lemma}
   \begin{proof} The former part is by induction on $n$.
     The latter is by structural induction (on the rules shown above).
   \end{proof}
   
   
   It is helpful to  define also another  auxiliary relation $\st \oreach \tau$ which holds if there exists $\mu$ satisfying $\reach \mu$ and $\tau \subseteq \term{\mu}$.
   This relation\footnote{It is easy also to give rules to define $\oreach$ inductively.} states that $\st$ reaches a set of terminal states. It is immediate that $\st \oreach \tau$ iff $\exists \rho$, $\st \redd^* \rho$ and $\tau \subseteq \term{\rho}$.

    \subsubsection{Properties of $\trsf$}     

 We now  study the action of $\trsf$ on transitions.  
 We first look at how $\trsf$ maps initial/final/deadlock states.

 \begin{lemma}\label{trsf_term}
   
   \begin{enumerate}
   \item If $\stII_{\R} \in [\trsf_{R_i}]$, then $\trsf_{R_i}(\stII_{\R}) = \stII_{\R_i}$ .
         
     
   \item  Assume  $\st\in [\trsf_{R_i}]$  is a final/deadlock state of  $\M_{R}$;
     then $\trsf_{R_i}(\st)$
     is a  final/deadlock state of  $\M_{R_i}$.
     
     
   \item If $\tau=\term{\tau}$ (\ie\ all states  are terminal), and  $\supp{\tau}\subseteq [\trsf_{R_i}]$,  then $\NF{\tau}= \NF{\trsf_{R_i} (\tau)}$.
   \end{enumerate}
   
 \end{lemma}

\begin{lemma}\label{lem:trsf}
  If $\st \smile \st'$ and $\st,\st'\in [\trsf_{R_i}]$, then $\trsf_{R_i}(\st) \smile \trsf_{R_i}(\st')$
\end{lemma}   
      
  %

    \newcommand{\stables}[1]{S(#1)}
  It is also important to understand the action of $\trsf$ on the number of stable tokens.
   We observe that the number of tokens, and stable tokens in particular, in any state $\st$ which is reached in a run of $\M_R$
     is finite.   We denote by $\stables {\st}$  the number of stable tokens in $\st$. 
   %
   The following is immediate by analyzing  the definition of transformation,  and checking  which tokens are deleted.
   \begin{fact}[stable tokens]\label{stable num}
    For any $\trsf_{R_i}$,  $\stables {\st} \geq \stables {\trsf_{R_i}(\st)}$.
      Moreover, if the reduction $\rednet$ is $d,y$ or $u_i$,
      then we also have that $\stables{\st} > \stables{\trsf_{R_i}(\st)}$.
   \end{fact}

\subsection{Invariance}\label{sec:invariance}
We prove the following result, from which invariance  (Th.~\ref{invariance}) follows.

\begin{proposition}[Main Property] \label{trsf_main} 
  Assume $\R \rednet \sum_i {p_i}\cdot \{\R_i\}$.
  $\stII_{\R}$ $q$-terminates  
  if and only if $\stII_{\R_i}$ $q_i$-terminates and $\sum (q_i \cdot p_i) = q$.
\end{proposition}

Let us first sketch the ingredients of the proof.  We need to work
  our way ``back and forth'' via Lemmas \ref{trsf_reach} and
  \ref{aux_lemma}, because of the following facts.
  \begin{varitemize}
  \item  
    Unfortunately, for $\stII_{\R} \redd^* \mu$  it is not true
    that $\trsf_{R_i} (\stII_{\R}) \redd^* \trsf_{R_i}(\mu)$.  However we
    have that if $\stII_{\R} \reach \mu$ in $\M_R$, then $\trsf_{R_i}
    (\stII_{\R}) \reach \trsf_{R_i}(\mu)$ (under natural conditions).  This is made precise by Lemma~\ref{trsf_reach}.
  \item  
    On the other side, the strength of the relation $\redd$ is that if
    $\stII_{\R}\redd^n \mu$, then for any sequence of the same length
    $\stII_{\R}\redd^n \rho$, we have that $\term{\rho}=\term{\mu}$.  This
    is not the case for the relation $\reach$ which is \emph{not
      informative}. The (slightly
    complex) construction which is given by Lemma \ref{aux_lemma} allows us to exploit the power of $\redd$.
          
  \end{varitemize}

\vskip 4pt
 
 We have everything in place to study the action of $\trsf$ on a run of the machine.    
 What is the action of $\trsf$ on a transition? By checking the definition in Fig.~\ref{fig:trsfOthers}
we observe that it may be the case that  $\st \redsiam \{\sttwo\} $ and  $  \trsf_{R\rednet R_i}(\sttwo)= \trsf_{R\rednet R_i}(\st)$.
We say  that such  a transition   \textit{collapses}
 for $\trsf_{R\rednet R_i}$. 
We observe some properties:

\begin{lemma}
  From a state of $\M_R$,
  we have at most a finite number of collapsing transitions.
\end{lemma}
\begin{proof}
  Since the reduction is surface, and since the type of any edge is finite,
  the set $\{(e,s,t) \,|\, e$ is an edge of the redex, $(e,s,t)$ is involved in a collapsing transition$\}$ is at most finite.
  Suppose there are infinitely many collapsing transitions from a state.
  Then there exist two or more tokens
  which have the same stacks involved in the sequence of transitions.
  They must have the same origin, and hence by injectivity
  they are in fact the ``same'' token visiting the redex twice or more.
  Therefore, by ``backtracking'' the transitions on that token,
  it again comes to the same edge in the redex with the same stack,
  hence we can go back infinitely many times.
  However this cannot happen in our \PSIAM{} machine,
  since any token starts its journey from a position in $\START$
  from which it cannot go back anymore, and transitions are bideterministic on each token.
\end{proof}

\begin{fact}
  Given a transition $\st \redsiam \mu $, if  $\st\in [\trsf_{R_i}]$, then either the transition collapses, or  $\trsf_{R_i} (\st) \redsiam  \trsf_{R_i}(\mu)$ is a transition of $\M_{R_i}$.
\end{fact}
  
\begin{lemma}\label{trsf_reach}
  If  $\st \in [\trsf_{R_i}]$ and $\st \reach \mu$ (in  $\M_R$),
  then  $\trsf_{R_i} (\st) \reach \trsf_{R_i} (\mu)$ holds.
\end{lemma}

 \begin{proof} 
  We transform a derivation $\Pi$ of  $\st \reach \mu$ in $\M_R$ into a derivation of $\trsf_{R_i} (\st) \reach \trsf_{R_i}(\mu)$ in $\M_{R_i}$, by induction on the structure of the derivation.

\newcommand{\UU}{\sttwo}

 \begin{itemize}
 \item Case 
 $
 \infer
 {\st {\reach}\{\st\}}
 { }
 \quad \text{becomes} \quad 
 \infer 
 {\trsf_{R_i} \st {\reach}\{\trsf_{R_i}(\st)\}}
 { }
 $
 
 \item Case 
 $
 \infer[]
 {\st{\reach}\sum p_{\UU}\cdot\mu_{\UU}}
 {\st{\redsiam}\sum p_{\UU}\cdot\UU & \{\infer{\sttwo{\reach}\mu_{\UU}}{\dots}\}}
 $\!\!

\noindent
 We examine the left premise, checking if it collapses:
 
 \begin{itemize}
 \item If it does not  collapse, $\trsf_{R_i} (\st){\redsiam}\sum p_{\UU}\cdot\trsf_{R_i} (\sttwo)$ is a transition of $\M_{R'}$ and we have:
 $$
 \infer
 {\trsf_{R_i}(\st){\reach}\sum p_{\UU} \cdot \trsf_{R_i} (\mu_{\UU})}
 {\begin{array}{c}\trsf_{R_i} (\st){\redsiam}\sum p_{\UU}\cdot \trsf_{R_i} (\sttwo)\\ \{   {\trsf_{R_i} (\sttwo) \reach \trsf_{R_i}(\mu_{\UU}) }  \} \text{ by I.H.}\end{array}}
 $$
 \item If it  collapses, we have   $\st \redsiam \{\sttwo\}$, we also have  $\trsf_{R_i} \st = \trsf_{R_i}(\sttwo)$,and
 the derivation $\Pi$ is  of the form: 
 $$
 \infer
 {\st{\reach}\mu}
 {\st{\red}\{\sttwo\} & \infer{ \sttwo{\reach}  \mu}
 {... }}
 $$
 
 By induction, $ \trsf_{R_i} \sttwo  \reach  \trsf_{R_i}(\mu)$, and
 therefore we conclude  $ \trsf_{R_i} \st  \reach  \trsf_{R_i}(\mu)$.
 \end{itemize}
 \end{itemize}
 
 %
 %
 %
 %
 %
 %
 %
 %
 %
  \end{proof}


 
   

 Lemma  \ref{trsf_reach}, the construction shown in Appendix~\ref{sect:construction}, and Lemma~\ref{trsf_term},
 allow us to transfer termination from $\stII_{\R}$ to  $\stII_{\R_i}$,
 and to prove one direction of Prop.~\ref{trsf_main}.
The other direction is more delicate. 
 
 Assume  that $\stII_{\R_i}$ $q_i$-terminates; this implies that for a certain $n$, whenever $\stII_{\R_i} \redd^n \sigma$ then 
 $\NF{\sigma}\geq q$. The  following Lemma builds 
  such a sequence in a way that  $\sigma=\trsf_{R_i}(\mu) $, with  $\mu$ in $\M_\R$. This allows us to transfer the properties of termination of $\stII_{\R_i}$ back to $\stII_{\R}$, ultimately leading  to the other direction of 
  Prop.~\ref{trsf_main}.

\begin{lemma}\label{aux_lemma}

  Assume  $\st \in [\trsf_{R_i}]$ .
   For any $n$:
\begin{enumerate}
\item  there exists $\mu$ such that  $\st\reach \mu$  and $\trsf_{R_i} (\st) \redd^n \trsf_{R_i}(\mu) $;
\item we can choose $\mu$ such that $\NF{\mu} = \NF{\trsf_{R_i}(\mu)}$.
\end{enumerate}
\end{lemma}

\begin{proof} 
 
\begin{enumerate}
\item  We build $\mu$ and its derivation, by induction on $n$.
 \begin{enumerate}
 \item[$n=1$.]
 
  \begin{itemize}
 \item Assume $\st$ is terminal, then $\trsf_{R_i} (\st)$ is terminal, and $\trsf_{R_i}  (\st) \redd \trsf_{R_i}  (\st)$.
 
 \item Assume  there is $\mu$ s.t. $\st \redone \mu$ non-collapsing. We have  $\trsf_{R_i}  (\st) \redd \trsf_{R_i}  (\mu)$.
 
 \item Assume that all transitions from $\st$ are   collapsing. For such a reduction, we have that $\st \redone \st'$ and  $\trsf_{R_i}  (\st) = \trsf_{R_i}  (\st') $. It is immediate to check that from any $\st\in \stsI{R}$ there is at most a finite number of consecutive collapsing transitions. We repeat our reasoning on $\st'$ until we find $\sttwo$
 which is either terminal or has a non-collapsing transition $\sttwo \red \mu$. The former case is immediate, the latter gives $\sttwo \reach \mu$ and therefore $\st \reach \mu$ by transitivity, and $ \trsf_{R_i}  (\st)= \trsf_{R_i}  (\sttwo) \red \trsf_{R_i}  (\mu) $, hence $\trsf_{R_i} (\st) \redd \trsf_{R_i}(\mu) $.
 \end{itemize}
 
 \item[$n>1$.] Assume we have built a derivation of $\st \reach \rho$ with

\noindent
 $\trsf_{R_i}(\st) \redd^{n-1} \trsf_{R_i} (\rho)$.
   We have that $\trsf_{R_i}(\rho) =
   \sum \rho(\sttwo)\cdot \trsf_{R_i} (\sttwo)$. For each $\sttwo\in \supp \rho$, we apply the base step, and obtain a derivation of $\sttwo\reach \mu_{\sttwo}$ with   $\trsf_{R_i}(\sttwo) \redd \trsf_{R_i}(\mu_{\sttwo})$. Putting things together, $\st \reach \sum \rho(\sttwo)\cdot \mu_{\sttwo}$ and  $\trsf_{R_i}(\st) \redd^n  \sum \rho(\sttwo)\cdot \trsf_{R_i} (\mu_{\sttwo})$ by Lemma~\ref{lem:conflict}.
  
 \end{enumerate}
 
\item We now prove the second part of the claim. Let  $\st \reach \mu$ be the result obtained at the previous point.
  Let $\{\sttwo_k\}$ be the set of states in  $ \supp{\mu}$ such that $\trsf_{R_i} (\sttwo_k)$ is terminal. This induces a partition of $\mu$, namely $\mu= \rho + \sum c_k \cdot \{\sttwo_k\}$.
  It is immediate to check  that each $\sttwo_k\reach \{\sttwo'_k\}$
  with $\sttwo'_k$ terminal and $\trsf_{R_i}(\sttwo'_k) = \trsf_{R_i}(\sttwo_k)$.
  Observe also that $\rho$ does not contain any terminal state.
  Let $\nu=\sum c_k \cdot \{\sttwo'_k\}$. We have by transitivity
 $\st \reach (\rho + \nu)$,  and $\trsf_{R_i}(\st) \redd^n \trsf_{R_i}(\rho + \nu) $
 (because $\trsf_{R_i}(\rho+\nu)= \trsf_{R_i}(\mu) $).  We have 
 $\NF{\trsf_{R_i} (\rho + \nu)} = \NF{\trsf_{R_i} (\nu)} = \sum c_k $ because $\trsf_{R_i} (\nu) = \sum c_k \cdot \trsf_{R_i}(\sttwo'_k)$. We conclude by observing  that $\NF{ \rho + \nu} = \NF {\nu}= \sum c_k$.
\end{enumerate}
 
 \end{proof}

%

%

 Summing up, we now have all the elements to prove Prop.~\ref{trsf_main}.
 \begin{proof}\textbf{(Prop.~\ref{trsf_main})}
 


 $\Rightarrow$. Follows from 
Prop. \ref{trsf_reach}, by using the  construction  in Sec.~\ref{sect:construction},  Lemma~\ref{trsf_term}, and  linearity of $\trsf$.

    Assume $\stII_R \redd^* \mu$, with $\term \mu$ not empty, and that 
    the machine starts as described  in Sec.~\ref{sect:construction} (in case $\rednet$ is \textit{link, update} or \textit{test}). We observe that
 every state $\st \in \supp{\mu}$ is contained in $[\trsf_i]$ for some $i$.
We  can then   prove that for each $i$ there exists $\mu_i \in \dists{\stsI{R}}$ such that
     $\stII_{R_i} \reach   \trsf_{R_i}(\mu_i)$,
    and  such that $\nu=\sum_i p_i \cdot \mu_i$. 

 $\Leftarrow$. Follows from  Lemma \ref{aux_lemma}. 
 We examine the only non-straightforward  case.
 Assume $\R\rednet_{\mathrm{test(i,\mm)}} \{\R_0^{p_0}, \R_1^{p_1}\}$.
 We choose a run of the machine which starts as described in Sec.~\ref{sect:construction};
 we have that $\stII_{\R} \reach \sum  p_i\cdot \{\st_i\}$,
 with $\trsf_{R_i}(\st_i)=\stII_{\R_i}$ by Lemma~\ref{lem:main}.
 By hypothesis, $\stII_{\R_i}$ terminates with probability at least $q_i$; assume it does so in $n$ steps. 
   By using  Lemma \ref{aux_lemma}, we build  a derivation  $\st_i \reach \mu_i$ such that  $\trsf_{R_i}(\st_i) \redd^n  \trsf_{R_i} ( \mu_i)$ and $\NF{\mu_i}= \NF{\trsf_{R_i} ( \mu_i)}$. By  Th. \ref{th:proba_term}, $\NF{\trsf_{R_i} ( \mu_i)}\geq q_i$.
  
  Putting all together, we have that $\stII_{\R} \reach \sum p_i \cdot \mu_i $, and  $\stII_{\R}$ terminates with probability at least $\sum p_i \cdot q_i$.
 \end{proof}

\subsection{\PSIAM\ Adequacy and Deadlock-Freeness: The Interplay of Nets and Machines}\label{sec:interplay}
\newcommand{\ww}[1]{\mathtt{weight}(#1)}

We are now able to establish adequacy (Th.~\ref{adequacy psiam}) and deadlock-freeness (Th.~\ref{deadlock-free}). Both are  direct consequence of Prop.~\ref{mutual} below, which in turn follows  form Prop.~\ref{trsf_main} and the following Fact,  
   by finely exploiting   the  interplay between nets and machine.

\begin{fact}\label{base interplay} Let $R$ be a  net of conclusion $\one$ and  such that  no reduction is possible. By Th.~\ref{cutel}, $R$ has no cuts, and is therefore simply a $\onelk$ node.  On such a simple net, $\M_R$ can only terminate in a final state: no deadlock is possible.
\end{fact}

\begin{proposition}[Mutual Termination]\label{net_termination}\label{mutual}
Let $\R$ be a  net of conclusion $\one$.\
The following are equivalent:
\begin{itemize}
\item[1.]
   $\stII_\R$ q-terminates;
\item[2.]

   $\R$ q-terminates;
\end{itemize}
Moreover \begin{enumerate}
\item[3.]  if $\stII_\R \reach \mu$ and $\st\in \supp{\mu}$ is terminal, then $\st$ is a \emph{final} state.
\end{enumerate}

\end{proposition}

\begin{proof}

\noindent 
\textbf{(1. $\Rightarrow$ 2.) and 3.}
We prove that
\begin{center}
 if  $\stII_\R\oreach \tau$, then 
 $
 \begin{array}{l}
        \mbox{(*) $\R$ terminates with probability at least $\NF{\tau}$,}\\  
        \mbox{(**) all states in $\supp{\tau}$ are final.}
     \end{array}
 $
\end{center}

 The proof is by double induction on the lexicographically ordered pair $(S(\tau), W(\R))$, where $W(\R)$ is the weight of the cuts \emph{at the surface} of $\R$, and  $S(\tau)=\sum_{\st\in \supp{\tau}} S(\st)$ with $S(\st)$ the number of stable tokens  in $\st$ (Fact\ref{stable num}). Both parameters are finite.

We will largely use  the  following fact (immediate consequence of the definition of $\oreach$ and of results we have already proved): if $\st \oreach \tau$ in $\M_\R$ and $\st \in[\trsf_i]$, then  $\trsf_i(\st) \oreach \trsf_i (\tau)$.


\begin{itemize}
\item If $\R$ has no reduction step, then $\NF{\R}=1$, which trivially proves  (*); (**) holds by Fact~\ref{base interplay}.
\item Assume  $\R\rednet_{\not \Mtest(i,\mm)} \R'$ (observe that this is a deterministic reduction). 
  We have that  $\stII_{\R'} \oreach \trsf (\tau)$, and $\NF{\trsf (\tau)}=\NF{\tau}$. 
  By Fact \ref{stable num}, $S(\trsf(\tau)) \leq S(\tau)$.
  If $\R\rednet_d \R'$, then  $S(\trsf(\tau)) < S(\tau)$.
  Otherwise $S(\trsf(\tau)) = S(\tau)$ but $W(\R') < W(\R)$
  because the step reduces a cut at the surface,
  and \emph{does not open any box}.
  Hence by induction, $\R'$ terminates with probability at least $\NF{\trsf(\tau)} = \NF{\tau}$ (and therefore so does $\R$) and  all states in $\trsf (\tau)$ are final, from which  (**) holds by Lemma~\ref{trsf_term}.2.
\item
Assume  $\R \mathrel{\rednet_{\Mtest(i,\mm)}} \sum p_i\cdot \{\R_i\}$. From $\stII_\R\oreach \tau$, by Lemma~\ref{lem:equiv} we have that there is $\rho$ satisfying  $\stII_\R\redd^* \rho$ and  $\tau \subseteq \term\rho$.
Using the construction in Sec.~\ref{sect:construction}, we have $\stII_\R\redd^* \sum p_i\cdot \{\st_i\}$, which induces  a partition of $\tau$ in $\tau= p_0\cdot \tau_0 + p_1\cdot \tau_1$ with $\st_i \oreach \tau_i$ for each $i$. 
 We have that   $S(\tau_i) < S(\tau)$, and that $\stII_{\R_i} \oreach \trsf_i(\tau_i)$, because $\trsf_i(\st_i)$ is defined and therefore
 $\trsf_i(\sttwo)$ is defined for each state $\sttwo \in \tau_i$.
By Fact \ref{stable num}, $S(\trsf_i(\tau_i)) \leq S(\tau_i) < S(\tau)$,
 thus by induction $\R_i$ terminates with probability at least
 $\NF{\trsf_i(\tau_i)}$, and all states in $\supp{\trsf_i(\tau_i)}$ are final.
 Therefore, $\R$ terminates with probability at least
 $\sum p_i \cdot \NF{\trsf_i(\tau_i)} = \sum p_i \cdot \NF{\tau_i} = \NF{\tau}$ by Lemma~\ref{trsf_term}.3,
 and all states in $\supp{\tau}$ are final by Lemma~\ref{trsf_term}.2.
\end{itemize}

\noindent
\textbf{2. $\Rightarrow$ 1. } By hypothesis, $\R\redd^n \rho$ with $\NF{\rho} \geq q$. We prove the implication by induction on $n$.

Case $n=0$. The implication is true by Fact \ref{base interplay}.

Case $n>0$. Assume $\R\rednet \sum p_i\cdot  \R_i$. By hypothesis,  each $\R_i$ terminates with probability at least $q_i$ (with $\sum p_i \cdot q_i = q$). By induction,  each $\stII_{\R_i}$ $q_i$-terminates, and therefore  (Prop.~\ref{trsf_main}) $\stII_\R$ $q$-terminates.
\end{proof}

\subsection{\PSIAM: Full development of Fig.~\ref{fig:recursion-msiam}}

  Here we fully develop what was sketched as description of the
  \PSIAM{} execution presented in Fig.~\ref{fig:recursion-msiam}.

  In the first
  panel (A), no box are yet opened: only two tokens are generated: the
  dereliction node emits token (a), in state
  $(*.\delta,\epsilon)^{\downarrow}$, while the $\mathsf{one}$-node
  emits token (b), in state $(\epsilon,\epsilon)^{\downarrow}$, and
  attached to a fresh address of the memory. Eventually token (a)
  reaches the entrance of the $Y$-box and opens a copy: its state is
  now $(\delta,*.\epsilon)^{\leftrightarrow}$. Token (b) also flows
  down: it first reaches the $\mathsf{H}$-sync node, crosses it while
  updating the memory, crosses the $\otimes$-node and gets the new
  state $(l.\epsilon,\epsilon)^{\downarrow}$. It continues through
  $\mathsf{?d}$ with new state $(*.l.\epsilon,\epsilon)$, reaches the
  $Y$-entrance: its copy ID is $*.\epsilon$, and it has been
  opened by token (a), it can carry onto the left branch with new
  state $(l.\epsilon,*.\epsilon)$. It arrives at the $\parr$-node and
  follow the left branch with state $(\epsilon,*.\epsilon)$: it now
  hits a bot-box.

  The test-action of the memory is called, and a probabilistic
  distribution of states is generated where the left and the
  right-side of the $\bot$-box are probabilistically opened: the
  corresponding sequences of operations are represented in Panel
  (B${}_0$) for the left side, and Panel (B${}_1$) for the right side.

  In Panel (B${}_0$): the left-side of the bot-box is opened and its
  $\mathsf{one}$-node emits token (c), in state
  $(\epsilon,*.\epsilon)$: note how the box stack of this newly-created
  token is the one of the copy of the $Y$-box it sits in. In any case,
  the token also comes equipped with a fresh address from the memory,
  and carries downward. When it reaches the entrance of the $Y$-box,
  coming from the left it exits and eventually reaches the conclusion
  of the net. Note how we end up with a normal form: token (b) and (a)
  are stable at doors of boxes.
  
  In Panel (B${}_1$): the right-side of the bot-box is opened and its
  $\mathsf{one}$-node emits a token, that we can also call (c), also
  in state $(\epsilon,*.\epsilon)$. The $\mathsf{?d}$-node emits token
  (d) in state $(*.\delta,*.\epsilon)$: this token flows down and gets
  to the entrance of the $Y$-box: it stops there in state
  $(\delta,y(*,*).\epsilon)$ and opens a new copy of the
  $Y$-box. Token (c) goes down, arrives at the $Y$-entrance and enters
  this new copy (of ID $y(*,*)$). It hits the corresponding copy of
  the $\bot$-box, and the test-action of the memory spawns a new
  probabilistic distributions.
  
  We focus on panel (C${}_{10}$) on the case of the opening of the
  left-side of the $\bot$-box: there, a new token (e) is generated
  (with a fresh address attached to it) and goes down. It will exit
  the copy of ID $y(*,*)$, enter the first copy, goes over the axiom
  node, and eventually exits from this first $Y$-box-copy. It is now
  at level 0, and goes to the conclusion of the net. The machine is in
  normal form.

\section{{\fontsize{10}{10}\selectfont\PCF{}}: Adequacy}\label{app:PCFproofs}

       The translation of \PCFAM{} closures into program nets is given  in
      Fig.~\ref{fig:PCFtoNET1} and~\ref{fig:PCFtoNET2}. There, we assume
      that the translation is with respect to the fixed pair
      $(\ind{},\mm)$. The partial map $\ind{R_M}$ is depicted with a dotted
      line, and corresponds exactly to the parameter to the map
      $\net{(-)}_{\ind{},\mm}$.

      \begin{figure*}
        \centering
        \scalebox{0.9}{\includegraphics[page=1,width=\textwidth]{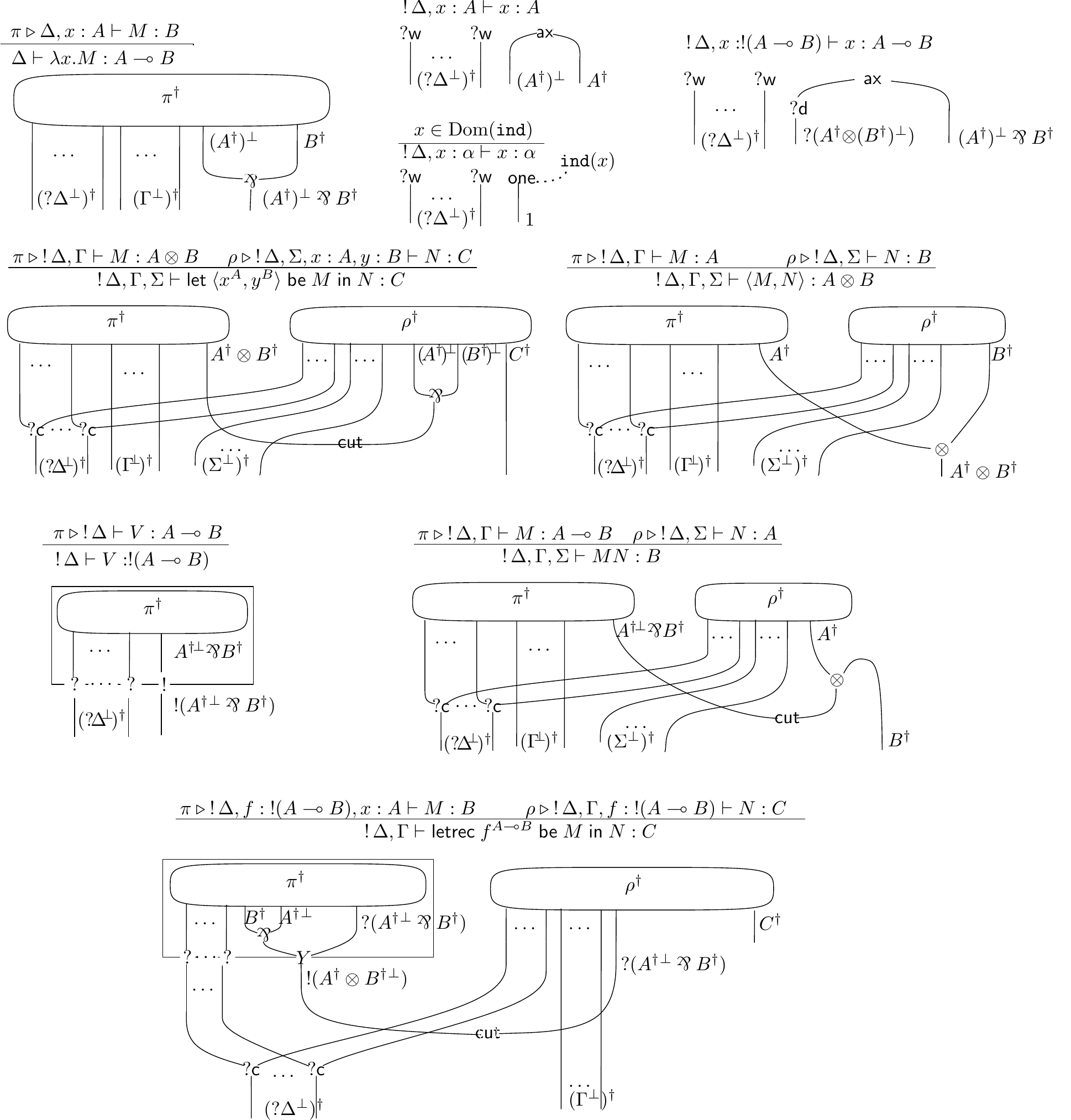}}
        \caption{Translation of \PCF\ into Nets.}
          \label{fig:PCFtoNET1}
      \end{figure*}
      
      \begin{figure*}
        \centering
        \scalebox{0.9}{\includegraphics[page=2,width=\textwidth]{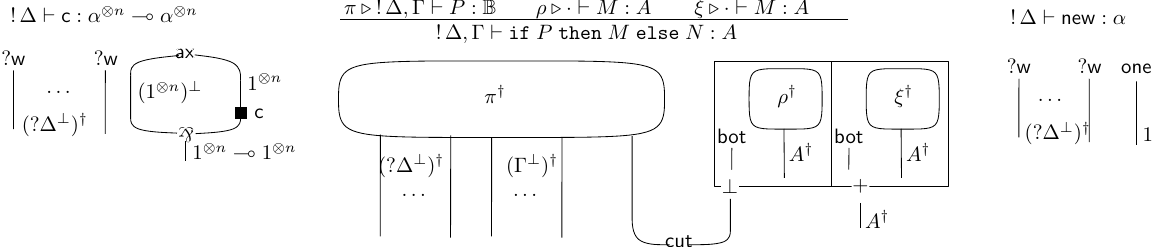}}
        \caption{Translation of \PCF\ into Nets.}
          \label{fig:PCFtoNET2}
      \end{figure*}

      With this definition, the well-typed closure
      $y_1:A_1,\ldots,y_m:A_m\PCFentail(M,\ind{},\mm):B$ can now simply be
      mapped to the program net $\net{M}_{\ind{},\mm}$.

      We prove here the adequacy theorem (Th.~\ref{th:adeqcbv}).

      \medskip
      \noindent
{\bf Th. \ref{th:adeqcbv} (Recall).}~~{\em
  Let $\PCFentail \PCFam:\PCFalpha$, then $\PCFam\,{\Downarrow_p}$
  if and only if $\PCFsemcbv{\PCFam}{\Downarrow_p}$.}
  
\medskip
Before proving the theorem, we first establish a few technical lemmas
which analyze the properties of the translation $\PCFsemcbv{(-)}$.
\begin{lemma}\label{lem:nf} 
  Assume that $\PCFam = (M,\ind{},\mm)$ is \PCFAM{} closure
  that $\PCFentail \PCFam:\PCFalpha$, and $\mu$ 
  a distribution of such closures. We have:
  \begin{enumerate}
  \item $\PCFam$ is a normal form \emph{if and only if}
    $\PCFsemcbv{\PCFam}$ is a normal form.
  \item $\NF{\mu} = \NF {\PCFsemcbv{\mu}}$.\qed
  \end{enumerate}
\end{lemma}

\begin{lemma}\label{lem:nonstuttering}
  Under the hypotheses of Lemma~\ref{lem:nf}:
  \begin{enumerate}
  \item if $\PCFam \PCFcbv \mu$ then
    $\PCFsemcbv{\PCFam} \redd^k \PCFsemcbv{\mu}$, with $k \geq 1$.
  \item if $\mu \redd^*\nu$ then $\net \mu \redd^* \net \nu$.\qed
 \end{enumerate}
\end{lemma}


\begin{corollary}\label{cor:ad} Under the hypotheses of Lemma \ref{lem:nf}:
\begin{enumerate}
\item  If  $\net\PCFam \rednet \rho$, then there is $\mu$ s.t. $\PCFam \PCFcbv \mu$ with  $\net\PCFam \not= \net \mu$. 
\item  If  $\net\PCFam \redd^k \rho$, then there is $\mu$ s.t.  $\PCFam \redd^* \mu$ and   $\net \PCFam \redd^m  \net \mu$, with $m\geq k$.
\end{enumerate}
\end{corollary}

\begin{proof}(1.) Immediate consequence of Lemma \ref{lem:nf} and \ref{lem:nonstuttering} (2.) By induction.
\end{proof}

We are now ready to prove Th.~\ref{th:adeqcbv}.

\begin{proof}[Proof of Th.~\ref{th:adeqcbv}]
Assume $\PCFam\Downarrow_{p_{\it term}}$ and $\net\PCFam \Downarrow_{p_{\it
    net}}$; we want to prove that $p_{\it term} = p_{\it net}$.

\noindent {$\mathbf{p_{term} \leq p_{net}}$.} It follows from the following. Assume $\PCFam\redd^*\mu$ with $\NF{\mu}=q$, then $\net\PCFam \redd^* \net \mu$ (by Lemma \ref{lem:nonstuttering}.2) and   $\NF{\net \mu}=q$ (by Lemma \ref{lem:nf}.2).

\noindent{$\mathbf{p_{term} \geq  p_{net}}$.}
We prove that if $\net\PCFam \redd^* \rho$ then it exists $\mu$ with  $\PCFam\redd^* \mu$ and   $ \NF{\mu} \geq \NF{\rho}$.
Assume $\net\PCFam \redd^k \rho$. By Corollary \ref{cor:ad}, $\PCFam\redd^* \mu$ and $\net\PCFam\redd^m \net \mu$, with $m \geq k$.   By Uniqueness of Normal Forms (Th. \ref{th:proba_term}.1) we have that  $\NF{\net \mu} \geq \NF{\rho}$. By Lemma \ref{lem:nf}, $\NF{\mu} = \NF{\net \mu}$, from which  we deduce the  statement.
\end{proof}

\end{document}